\documentclass[10pt, conference, letterpaper]{IEEEtran}

\usepackage{tikz}
\usepackage{dirtytalk}

\usepackage{xcolor}
\usepackage{amsmath}
\usepackage{amsthm}
\usepackage{amssymb}
\usepackage{amsfonts}
\usepackage{graphicx}
\usepackage{subfigure}
\usepackage{url}
\usepackage{booktabs} % for much better looking tables
\usepackage{array} % for better arrays (eg matrices) in maths
\usepackage{verbatim} % adds environment for commenting out blocks of text & for better verbatim
\usepackage{caption}
\usepackage{enumitem}
\usepackage{bbm}

\usepackage{adjustbox}
\usepackage[makeroom]{cancel}
\usepackage[font=footnotesize]{caption}

\newtheorem{lemma}{Lemma}

\newtheorem{property}{Property}

\usepackage{amsthm}
\usepackage{hyperref}
\usepackage{cleveref}

\makeatletter
\def\@ythm#1#2#3[#4]{\def\@currentlabelname{#4}%
  \expandafter\global\expandafter\def\csname#1name\endcsname{#4}%
  \@opargbegintheorem{#3}{\csname the#2\endcsname}{#4}%
  \ifx\thm@starredenv\@undefined
    \thm@thmcaption{#1}{{#3}{\csname the#2\endcsname}{#4}}\fi
  \ignorespaces}

\makeatother

\newtheorem*{problem}{}

\newtheorem*{corollary}{Corollary}

\usepackage{tcolorbox}
\usepackage{color}

\usepackage{soul}
\setstcolor{red}

\newcommand{\oldtext}[1]{}

\newcommand{\theo}[1]{\textcolor[rgb]{1.00,0.00,1.00}{[{#1} --theo]}}

\newcommand{\myitem}[1]{\vspace{0.25\baselineskip}\noindent\textbf{#1}}%\vspace*{0.04in}
\usepackage{comment}
\usepackage{algorithm}
\usepackage{algorithmicx}
\usepackage{algpseudocode}

\def\BibTeX{{\rm B\kern-.05em{\sc i\kern-.025em b}\kern-.08em T\kern-.1667em\lower.7ex\hbox{E}\kern-.125emX}}

\hypersetup{draft}

\begin{document}

\newfloat{subroutine}{htbp}{loa}
\floatname{subroutine}{Subroutine}
%\title{Optimal look-ahead policies for flexible\\ network-friendly recommendations} 
% \title{A flexible MDP framework for\\ network friendly recommendations}

\title{SOBA: Session optimal MDP-based network friendly recommendations}

\author{\large{Theodoros Giannakas\textsuperscript{1},
 Anastasios Giovanidis\textsuperscript{2},        
 and Thrasyvoulos Spyropoulos\textsuperscript{1}
 }\\
 \normalsize
 \textsuperscript{1}~EURECOM, Sophia-Antipolis~France, first.last@eurecom.fr\\
 \textsuperscript{2}~Sorbonne University CNRS-LIP6, Paris, France, anastasios.giovanidis@lip6.fr 
 }

\maketitle
\thispagestyle{plain}
\pagestyle{plain}
\begin{abstract}
Caching content over CDNs or at the network edge has been solidified as a means to improve network cost and offer better streaming experience to users. Furthermore, nudging the users towards low-cost content has recently gained momentum as a strategy to boost network performance. We focus on the problem of optimal policy design for Network Friendly Recommendations (NFR). We depart from recent modeling attempts, and propose a Markov Decision Process (MDP) formulation. MDPs offer a unified framework that can model a user with random session length. As it turns out, many state-of-the-art approaches can be cast as subcases of our MDP formulation.
Moreover, the approach offers flexibility to model users who are reactive to the quality of the received recommendations. In terms of performance, for users consuming an arbitrary number of contents in sequence, we show theoretically and using extensive validation over real traces that the MDP approach outperforms myopic algorithms both in session cost as well as in offered recommendation quality. Finally, even compared to optimal state-of-art algorithms targeting specific subcases, our MDP framework is significantly more efficient, speeding the execution time by a factor of 10, and enjoying better scaling with the content catalog and recommendation batch sizes.

% The approach offers great flexibility to model users who are reactive to the quality of the received recommendations. 
% Overall, for users who are requesting multiple contents in sequence, we show theoretically and using extensive validation over real traces that the MDP approach outperforms candidate heuristic algorithms both in session cost as well as in offered recommendation quality.
% Finally, we consider the execution time needed for the MDP policies. 
% Our algorithm is computationally more efficient compared to recent approaches and scales well with the content catalogs and the size of the recommendation batch.

% In the final section, we face the challenging problem of users that are reactive to the quality of recommendations. Along these lines we design a suboptimal and an $\epsilon$-optimal policy both with runtime guarantees. Importantly, the generality it displays along with its complexity guarantees, make the MDP a very appealing solution. That is due to the fact that the proposed method is to be deployed on top of the existing infrastructure and is essentially a software solution.
\end{abstract}

%% Section 1
\section{Introduction}
\label{sec:intro}
\subsection{Motivation}
With multimedia traffic from Netflix, YouTube, Amazon, Spotify, etc. comprising the lion's share of Internet traffic~\cite{cisco2015}, reducing the ``cost'' of serving such content to users is of major interest to both content providers (CP) and network operators (NO) alike. This cost includes the actual monetary cost for the CP to lease or invest in network and cloud resources, but also network-related costs, related to resource congestion, slowing down other types of traffic, stalling multimedia streams etc. 

Caching popular content near users has been a key step in this direction in wired networks through the use of CDNs~\cite{farber2003internet}, and more recently in wireless networks through femtocaching~\cite{femto}. In addition to cost reduction for CPs and NOs, caching also allows for higher streaming rate, shorter latency, etc.~\cite{doan2018tracing}, which results in an improved viewing/listening experience for the user. When platforms of video streaming services can not offer high bitrate, user abandonment rates rise \cite{nam2016qoe}. Hence, reducing the cost of bringing interesting content to users will benefit everyone: the users, the content providers, and the network operators.

%Therefore, to ensure a long-term engagement of the user to the platform, content providers (CPs) need to pay attention to the overall quality of experience (QoE) of a delivered content. It is increasingly understood that this QoE largely depends on \emph{both} the interest of the user in the content delivered (a minimum assumed requirement for any CP) \emph{and} the QoS of the actual delivery~\cite{RecExperiment}.  

Recommendation systems (RSs) in popular content platforms play an important role for this task: they suggest interesting content to users. For example, $80\%$ of requests in Netflix, and more than $50\%$ on YouTube, stem from the platform recommendations~\cite{gomez2016netflix,RecImpact-IMC10}. The traditional role of an RS has been to make personalized recommendations to the user, suggesting items from a vast catalog that best match her interests using techniques like collaborative filtering~\cite{sarwar2001item}, deep neural networks~\cite{covington2016deep}, matrix factorization~\cite{koren2009matrix}, etc. The vast majority of popular RS systems focus on content relevance and similarity, but they do not account for the network cost of delivery. Such operation, which ignores network costs in content recommendation algorithms inevitably leads to largely sub-optimal network performance for all parties involved.
%However, the QoS, e.g., where the content is cached and whether this allows the recommended content to be delivered at low or high streaming quality, does not usually factor into the decision of most recommenders.

\subsection{Related Work}

 A handful of recent works have spotted the interplay between recommendation-network vs QoS-cost, and have proposed to modify the recommendation algorithms towards a more \emph{network-friendly} operation~\cite{sermpezis2018soft, giannakas2018show, content-recommendation-swarming, chatzieleftheriou2019TMC, song2018making, kastanakis-cabaret-mecomm, liu2018learning}. The main objective of almost all these works is to recommend content that is highly interesting to the user while at the same time involves low delivery cost. A simple solution to achieved this is \emph{to favor cached content}~\cite{cache-centric-video-recommendation}. While various implementation barriers are sometimes cited~\cite{al2019qos3}, the increasing convergence of CPs and NOs~\cite{krolikowski2018optimal}, especially in the context of network slicing and virtualization suggests, that in the very near future content providers will be the owners of their own network (slice), and will be able to directly infer the potential network cost of recommending and delivering some content versus another. 

%\textbf{AKIS: the second half of the previous paragraph must definitely go somewhere in the paper; people often complain about this, so it should be clear that we assume the CP itself will implement our algorithm, and DOES NOT need any extra collaboration or network info from the NO, besides what is already assumed to be available for the owner of a 5G/5G+ network slice. It could though also be a footnote, to not break the flow of the story.}  

To date, a number of these early network-friendly RS proposals are basically (sometimes efficient) heuristics~\cite{chatzieleftheriou2019TMC, kastanakis-cabaret-mecomm}. A large number of these works focuses on \emph{myopic} algorithms, where the RS aims to minimize
%ations are modified assuming that each user requests only one content or towards minimizing 
delivery cost only for the next content request~\cite{cache-centric-video-recommendation}. In practice, however, when visiting popular applications like YouTube, Vimeo, Spotify, etc., a user
~\cite{cioYoutubeSessions, businessYoutubeSessions} consumes several contents one after the other, guided and impacted by the RS system at each of these steps. As a result, what the RS recommends while the user watches some content in a viewing session, will not impact the selection and delivery cost of just the next request, but also 
%the $(k+1)$-th, but can further impact the selection and delivery cost of 
all subsequent requests until the end of that session.

Myopic schemes are thus sub-optimal. Instead, one should aim to find the optimal action now, that will \emph{minimize the expected cost over the entire session, taking into account both what the RS could suggest in future steps, as well as how the user might react to them}. A couple of recent works have attempted to tackle this exact problem using convex optimization~\cite{giannakas2018show, giannakas2019wiopt}. While the authors manage to formulate the problem as a biconvex~\cite{giannakas2018show} and linear program~\cite{giannakas2019wiopt}, the latter yielding an optimal solution, these works are characterized by the following two key \emph{shortcomings}: (i) the problem formulation requires the user session to be of infinite length in order to derive closed form expressions for the objective; (ii)  
although the problem is an LP in~\cite{giannakas2019wiopt}, we will see that the runtime of their algorithm is quite slow.
%the user model is simplistic, assuming that the user clicks on recommended content irrespective of the content's relevance to her interests. 

\subsection{Contributions and Structure}
In this work, we approach the above Network Friendly Recommendations (NFR) problem in a novel way.
%by introducing a Markov Decision Process (MDP) formulation; we derive optimal recommendation policies that minimize the expected cost over arbitrary length user-sessions, while also offering content of relevance to the user either implicitly or explicitly, and in a much larger range of settings. 
Our main contributions can be summarized as follows:

% \textbf{(1)} We propose a unified MDP framework to minimize the expected caching cost over a user session, while offering content of relevance to the user. This allows the RS to deal with sessions of arbitrary length (myopic, short-term, or long-term) and incorporate various assumptions on the user's reaction to the quality of recommendations. 
%Our formulation can include several stricter problems that have been treated in the past literature and also solve new interesting cases.

\textbf{(C.1)} We propose a unified MDP framework to minimize the expected caching cost over a \emph{random session of arbitrary length}, while suggesting the user high quality content. Our approach is parameterized in such a way that many state-of-the-art methods be adapted to our framework.

\textbf{(C.2)} The MDP is formed using as problem unknowns the continuous \emph{item recommendation frequencies} per viewed content. In doing so, we do not need to search for the optimal $N$-sized recommendation batch per viewed content, thus avoiding the curse of dimensionality of MDPs.
%practically, it would be computationally intractable to search for the best among all possible N-size batches from a large content catalog of size K, see also \cite{slateQ}. To illustrate this, for $N=3$ recommendations and $K=1000$ there are more than $165$ Million options per viewed content. 
Our formulation uses the least number of variables ($K^2$ specifically, where $K$ is the size of the catalog) to describe the MDP without losing in optimality, compared to the fully detailed description. Noteworthy is the fact that 
the complexity of the algorithmic solution becomes insensitive to the size $N$ of recommended batch per viewed content.
% the complexity of the algorithmic solution becomes insensitive to the number $N$ of recommendation items.

\textbf{(C.3)} We express the content transition probabilities in a general way, which enables us to incorporate a variety of user behaviors. Furthremore, the policy iteration steps in the solution of the Bellman equations can be naturally decomposed into simpler continuous subproblems, which can be solved (i) by low complexity linear, or convex programming techniques 
%(depending on how the recommendations affect the content transitions) 
and (ii) in parallel, offering an additional potential speedup of $K \times$.

% \textbf{(4)} The benefit of the optimal MDP policy over simpler heuristic myopic policies
% %, which do not take the future into account. The main difference lies in the fact that the optimal MDP policy over the user session will
% is that when two contents are seemingly \emph{identical}, i.e., both low cost and relevant for the user, MDP can choose the right one in order to decrease the expected cost, by keeping the user happy at all times.

% For sessions with big horizon, the MDP has significant gains in terms of caching cost over myopic approaches that consider the NFR problem. 
% % For sessions with smaller horizon, it convincingly beats state-of-the-art methods that can only optimize over an infinite horizon. 
% Importantly, when compared to recently published works that consider the infinite horizon NFR problem, our framework (due to the reasons mentioned in (2)) is able speed-up the execution time by a factor of 10, while achieving a $\epsilon$-optimal cost performance.

\textbf{(C.4)}
For sessions with big horizon, the MDP has significant gains in terms of caching cost over myopic policies. When compared to recently published works that consider the infinite horizon NFR problem, our framework (due to the reasons mentioned in (C.2) and (C.3)) is able speed-up the execution time by a factor of 10, while achieving %$\epsilon$-
optimal cost performance.

The paper is structured as follows. Section II sets up the problem, presents the user-RS interaction and introduces the RS input and objectives. Feasible, optimal and sub-optimal recommendation policies are discussed. The problem is formed as an MDP in its general form in Section III and an algorithmic solution is described based on Bellman equations. In Section IV we present our user model and explain its MDP solution. Section V contains the evaluation of our policy in terms of cost and user satisfaction performance against heuristics and state of the art solutions. We conclude the paper in Section VI.

\section{Problem Setup}
\label{sec:problem-setup}
\subsection{User session and recommendations}

We consider a user who enters some multimedia application, e.g. YouTube, and requests sequentially a random number of items from its catalog $\mathcal{K}$ ($|\mathcal{K}|=K$). Such applications are equipped with a RS, responsible for helping the users discover new content. Our user has some prior underlying probability to request content $i$ from $\mathcal{K}$, which we denote as $p_0(i)$; with vector $p_0$ denoting the probability mass function (pmf) for all $i\in\mathcal{K}$.
% $p_{0}(i)$, and $\mathbf{p}_0$ is the probability mass function (pmf).
%The user starts off by randomly choosing an item among the library content, following a probability distribution $\mathbf{p}_{0}$, which expresses her personal preferences over the library. Once she clicks on it and starts viewing, the RS further suggests $N$ new contents, which appear on the side of the screen. We denote the vector of the $N$-sized recommendation batch as $w$. Every time she clicks on some new content the RS suggests $N$ different new ones.
%Depending on her behavior and how she assesses the offered contents, the user might click on one of these contents, or else use the text-bar to search for an alternative one drawn again with probability $p_0$. The user keeps on clicking until eventually she quits the application, signifying the end of the session.
% Depending on his behavior and how he assesses the suggested contents, the user might click on one of them, or search for a new content out of $\mathbf{p}_0$. After having selected a new content, a new set $N$ suggestions will be offered to him by the RS and so on.
% \myitem{Recommendation System}
% The RS is the mechanism which is responsible for suggesting $N$ new contents to the user application's screen when he clicks on some content $i \in \mathcal{K}$.
%\myitem{Session Length.}
The length of each session is random, and we assume that it follows a geometric distribution with mean $1/(1-\lambda)$. It is further assumed here that the session length is \emph{independent} of the RS suggestions.
%\myitem{User Interaction with the Application.}
%In summary, while viewing content $i$, the user has three possible actions:
The user session has the following structure:
\begin{itemize}
\item The user starts the session from some random content $i$ drawn from the distribution $\mathbf{p}_0$.
\item The RS, at every request, recommends $N$ new contents; we denote this $N$-sized batch as $w$.
\item The user may follow the recommendations related to content $i$, by clicking on a content among the $N$ in the batch $w$,
\item Or the user ignores the recommendations and chooses some other item based on initial preferences $p_{0}(i)$.
\item The user exits the session with probability $1-\lambda$ after any request.
\end{itemize}

%\end{comment}

% end Akis version

% Importantly, $\alpha_{ij}$ expresses the probability that the user will click on content $j$ when in $i$ given that it appeared on the recommendation batch, i.e., $\alpha_{ij} = P(\text{click on}~ j\mid j \in w)$. In Section \ref{sec:models}, we will discuss different modeling approaches for $\alpha_{ij}$.

\subsection{System input about user preferences} \label{subsec:inputs}
% \theo{$u_{ij}$ maybe interpreted as (1) historical average observed transitions or (2) item to item similarities, note that in the case of youtube there are no ratings such as we could speak for item-item similarity}
Entertainment oriented applications massively collect data related to user interaction and content ranking, allowing them to become increasingly effective in their recommendations.
% Contents are requested, rated/liked and accepted or rejected as part of a recommendation batch.
% ; all these feedback measurements can only be an added to RS and its objectives.
According to the RS literature \cite{herlocker1999algorithmic}, \cite{sarwar2001item}, \cite{crovellaOpinion} user ratings are used to infer the level of similarity \cite{LvNews} between contents.
% , that is the first stage of a collaborative filtering method. 
In our paper, we formalise the notion of \emph{related content} to viewed content $i$ as follows: 
%\myitem{Content Relations Graph.}
For every content $i$, there exists a similarity value with all other items in the catalog $\mathcal{K}$. The similarity with content $j$ is quantified by the value $u_{ij}\in [0,1]$, forming the $K$-length row vector $\mathbf{u}_i$. This information is summarized in the square non-symmetric $K\times K$ matrix $U$. 
%We set the diagonal of $U$ to 0, to ensure that only content(s) j \ne i are recommended when viewing i" 

We further denote by $\mathcal{U}_i(N)$, the set with the $N<K$ highest $u_{ij}$ values related to content $i$. 
%We will refer to the sum of the $N$ best possible $u_{ij}$ values for item $i$ as $Q_{i}^{max}$.
Note that the values of $\mathbf{u}_i$ are not normalized per content, i.e. the matrix $U$ is \emph{not} stochastic. The matrix $U$, which represents the content relations, is considered as \emph{input} for the RS.
%Using the above, the content library can be described as a directed graph having as vertices the library items and edges between related nodes, weighted by the non-null elements of $\mathcal{U}$. Hence, $i$ has outgoing edges for the contents in $\mathcal{S}_{i}$ with corresponding weights $\mathcal{U}_{i}$.

The RS assumes that the user feels satisfied with some recommendation batch, if this includes items $j$ with high $u_{ij}$ values.
%.user feels happy when her recommendation batch $w$ for content $i$ includes items with high $u_{ij}$ values.
% Essentially this quantity refers to how happy the user feels by looking at the whole recommendation batch.
\textit{User satisfaction} is denoted by $Q_{i}(w)$ and is quantified by the ratio 
% T$Q_i(w):=\underset{l \in w} \sum u_{il}$. 
% Consequently, the user satisfaction at content $i$ is 
\begin{equation}\label{eq:user-satisfaction}
Q_{i}(w) := \underset{j \in w} \sum u_{ij} \bigg/ \underset{m \in \mathcal{U}_{i}(N)} \sum u_{im}.
\end{equation}
It is measured per viewed content $i$ and recommendation batch $w$; it depends on the entries of $U$, the size $N$ of the batch and the policy. The denominator in (\ref{eq:user-satisfaction}) is the maximum batch quality $Q_{i}^{max}$ so that  $Q_{i}(w)\in [0,1]$. The expression states that the higher the sum $u_{ij}$ of the recommendation batch, the happier the user is. 
%Note that different content batch $w$ will cause a different user satisfaction at item $i$.
%, where $Q_{i}^{min}$ is the sum of the $N$ lowest $u_{ij}$ entries. 
% As a consequence, the maximum attainable quality we can get for content $i$ in the catalogue $\mathcal{K}$ is then defined as \begin{equation}\label{qmax}
% q_i^{max} = \sum_{l \in \mathcal{U}_{i}(N)} u_{il}
% \end{equation}
%
%Additionally, the RS has at its disposal statistics over the history of aggregate user requests and can estimate global content popularity $\mathbf{p}_{0}$. We assume that each content $i$ can be requested independently of the RS suggestions, just through the search bar of the application, with probability $p_{0}(i) > 0,~\forall~i~\in~ \mathcal{K}$. This probability distribution is the second RS input. 
Both the content popularity vector $\mathbf{p}_0$ and the similarity matrix $U$ is information that the RS has at its disposal, from measurements over time.

\subsection{Network-related costs}
From the network's perspective, each content $i \in \mathcal{K}$ has a non-negative network cost $c_{i},~\mathbf{c} = [c_1, \dots, c_K]^{T}$, associated to its delivery to the user. The content delivery cost might depend on several factors such as its size (in MB), its routing expenses, its location on the network etc. A session of $L>1$ requests incurs a cumulative cost on the network. 
Due to the impact of RS on user requests, the sequence of costs $\{c(S_{t})\}_{t=0}^{L}$ will depend on the RS policy, where $S_t$ is the state visited at $t$ and $c(S_{t})$ is the cost of the state $S_t$. Thus, our primary objective is to come up with policies $R$ (to be defined more formally next), which promote low-cost contents and ultimately minimize the session's average cost, while at the same time satisfying the user's natural preference for higher content relevance
%
%\theo{Allaksa to total se mean, gt auto kanoume. tonizw edw gia to last pass. To kana gt thelw na kanw refer stin poly vasiki sxesi gia to sim section}
% We consider a CDN which is responsible for delivering content to the users. Depending on routing, size and location of the content, the CDN has a nonnegative cost $c_{i}$ associated with the request for every content $i \in \mathcal{K}$.
% Thus when a user has a session of $M$ requests, his session \emph{incurs} some cost to the CDN.
% Importantly, as the RS policy can shape the content demands, this sequence of costs $\{c(t)\}_{t=0}^{M}$ heavily depends on the RS policy $\boldsymbol{\pi}$.
% The RS main objective is to minimize the session's cost by promoting low-cost content.
% Nevertheless, note that as the requests are affected by the RS policy, the cost of a request sequence $\{c(t)\}_{t=0}^{M}$ can also be driven by the RS policy $\boldsymbol{\pi}$. Hence, the RS \emph{main objective} is to act as the middleman and decide a policy $\boldsymbol{\pi}$ which
\begin{align}
\underset{R}{\textnormal{minimize}}~\bigg\{ \frac{1}{L}\sum_{t=1}^{L} c(S_{t})\bigg\}. \label{eq:objective}
\end{align}

% Our main objective is to find a recommendation policy for this library which will minimize the cost of a user session
% Our setup involves a set of users $\mathcal{M}$, that are connected wirelessly to a base station (BS) in an urban area. The BS is connected to the Internet through the backhaul, but at the same time is equipped with a finite storage space. Then according to the estimations of content request intensity, operators select a set of contents, which they upload overnight to the BS and leave it for the next day in order to serve users with better rates.

Letting $c_i \geq 0$ to be real positive, gives the flexibility to capture various network-related scenarios such as
\begin{enumerate}
    \item Maximize cache-hit rate: set $c_{i} \in\left\{0,1\right\}$ for $\left\{cached,\ uncached\right\}$ contents respectively.
    \item Minimize content delivery cost: set $c_{i} \in \mathbb{R}$ to include delay and bandwidth in the CDN case.
\end{enumerate}

\subsection{Policies} \label{subsec:action-space}
Our focus in this work is to find policies for arbitrary user sessions in terms of average length. The policies should aim at minimizing the expected session network-cost, while guaranteeing a good (and controllable) level of user satisfaction. Before formally defining the optimization problem in the next section, we present here in detail what is a policy and how it is modeled in our framework, and also three reasonable heuristics. % and we will argue why optimisation with look-ahead is more promising than myopic approaches.
% However, as pointed out in the previous two subsections, while our main objective is to minimize the cost of the user's session, the user satisfaction remains fundamentally an important dimension of the RS. Before going any further, we must first precisely define what we mean by policy.
As mentioned previously, when the user visits file $i$, the RS proposes any $N$-sized recommendation batch of unique contents (excluding self-recommendation $i$). The set of all $N$-sized batches $w$ forms the set of actions when the user views content $i$, which we denote as $\mathcal{A}_i$. To formally define a \textit{policy}, we need to associate each recommendation batch $w\in\mathcal{A}_i$ with a frequency of use $\mu_{i}(w)$. The frequencies of all the batches related to $i$ should sum up to 1.
%Thus, we have in total $\binom{K-1}{N} \cdot K$ frequencies to define. 
This gives rise to two classes of policies. 
\begin{itemize}
\item \emph{Deterministic:} A unique batch $w$ can appear per viewed object. For every $i$ there is a single action $w$ for which $\mu_{i}(w) = 1$.
\item \emph{Randomized:} At least two actions have $\mu_{i}(w) > 0$. This means that at every appearance of $i$, the user might see a different $N$-tuple of contents, chosen randomly.
\end{itemize}
The cardinality of the action set $\mathcal{A}_i$ is exploding, leading to $\binom{K-1}{N}$ variables \emph{per item} over which we must optimize.
%When using this brute force approach of finding the best batch to recommend among all possible ones, we face the modeling and computational challenge of having to enumerate all possible batches and introduce a frequency variable per batch per viewed content. 
As an example, for catalog $K = 1000$ and $N = 3$ recommendations the RS needs to introduce 165 Billion unknown $\mu$'s. 

\subsubsection{Item-wise recommendation frequencies} To overcome this serious modeling issue, we use a different approach. Related to viewed content $i$, we introduce the item-wise recommendation frequencies $\mathbf{r}_i=\left\{r_{ij}\right\}$ as the new set of unknown variables. In fact, these quantities can be expressed through the per-batch frequencies, and they actually summarize their information as follows, 
\begin{align}\label{eq:item-batch}
r_{ij} = \sum_{w \in \mathcal{A}_{i}} \mu_{i}(w)\mathbf{1}_{\{j \in w\}}~,\ \forall j \in \mathcal{K}.
\end{align}
%The above reads as follows, \say{for all possible actions in $\mathcal{A}_i$, sum the frequencies of those batches $w$ which recommend $j$}. 
Therefore, $r_{ij}\in[0,1]$ represents the overall probability of object $j$ to appear in any recommendation batch related to $i$, without specifying the other $N-1$ elements of the batch. For the vector $\mathbf{r}_{i}$ we can verify that it satisfies the size $N$ of the recommendation batch, with equality

\begin{align}\label{eq:sumrij}
    &\sum_{j=1}^{K} r_{ij} = \sum_{j=1}^{K} \sum_{w \in \mathcal{A}_{i}} \mu_{i}(w)\mathbf{1}_{\{j \in w\}} = N~\forall i \in \mathcal{K}.
\end{align}

If the policy is deterministic, then for every content $i$ there are exactly $N$ entries $r_{ij} = 1$, and the rest are equal to zero. 
On the other hand, if the policy is randomised, then at least two entries $r_{ij}<1$. To see this in a small example, consider the randomised policy with feasible batches $\mathcal{A}_{i} = \left\{\{1,2\}, \{1,3\}\right\}$ associated with batch-frequencies $\{0.5,\ 0.5\}$. This translates to item-wise frequencies $r_{i1} = 1.0,\ r_{i2} = 0.5$ and $r_{i3} = 0.5$, while the remaining $r_{ij}$'s are zero.
% $\mathcal{A}_{i} = \{1,2\}, \{1,3\},\{3,4\}$ associated with frequencies $[0.5,0.25,0.25]$. This translates to $r_{i1} = 0.75, r_{i2} = 0.50, r_{i3} = 0.50, r_{i4} = 0.25$. The remaining elements of $\mathbf{r}_i$ are zero, and as we can see, the nonzero $r_{ij}$ sum to $N$.
% To see why the above holds, we need to change the summation order, then for some \emph{fixed action} $w$, we sum $K$ indicator functions where by the action definition, only $N$ $j$'s out of the $K$ will belong to the action set.
For each content $i$, we relate a frequency vector $\mathbf{r}_i$ of size $K$. By concatenating these vectors as $R = [\mathbf{r}_{1}^{T},...,\mathbf{r}_{K}^{T}]$ $\in \mathbb{R}^{K \times K}$ we form the \emph{policy}. We have thus reduced the unknowns to just $K^2$, a considerable improvement!

\textbf{Remark}: The definition of a policy $R$ through the $r_{ij}$ frequencies, can allow to generate recommendation batches with the appropriate $\mu_{i}(w)$ batch-frequencies. For a deterministic policy, the $N$ non-zero $r_{ij}$ entries per $i$ define the unique $N$-sized batch $w\in\mathcal{A}_i$. Now, in the case of a randomised policy, for some $j$'s it holds $r_{ij}<1$, so there are more than one potential batches. We can use the random vector generation technique found in \cite[Fact 1, Probabilistic Placement Policy]{blaszczyszyn2015optimal}, where different batches of size $N$ are randomly sampled, while guaranteeing that each content $j$ appears with probability $r_{ij}$. 
%To enumerate all possible batch-actions, we can pick the different $N$-sized combinations in \cite[Fig.1]{blaszczyszyn2015optimal} and determine the probability of a specific batch $w$, by its width.
In the case of our previous simple example given $r_{i1} = 1.0, r_{i2} = 0.5$ and $r_{i3} = 0.5$, we can reproduce the batches and their frequencies as follows. Given $N=2$ recommendation slots, each slot will be time-shared by contents whose frequencies sum-up to $1$. So the first slot will always be occupied by item $``1"$ because $r_{i1} = 1.0$. The second slot will be time-shared by $``2"$ and $``3"$, $50\%$ of the time each, so that $\mu({\{1,2\}})=0.5$ and $\mu({\{1,3\}})=0.5$, thus reproducing the more detailed policy. This technique can be generalised to $N>2$.

\subsubsection{Simple Myopic Policies} Here we list some practical intuitive policies, which either favor low network-cost or user satisfaction or both, but are myopic in the sense that they consider only the impact of the immediate next request.
\begin{itemize}
    \item \textbf{Top-$N$ policy ($R_{Q}$):} Suggest the $N$ files that are most similar to $i$, i.e. the ones that correspond to the similarities in $\mathcal{U}_{i}(N)$ (ties broken uniformly). This choice maximizes user satisfaction.
    \item \textbf{Low Cost policy ($R_{C}$):} Suggest the $N$ contents with lowest cost. In the case of ties for the cost $c_{j}$, recommend contents arbitrarily.
    \item \textbf{$q$-Mixed policy ($R_{MIX}$):} Assign $q \cdot 100$ of the budget for user satisfaction. If items are tied in $u_{ij}$, choose the lowest cost to favor the network. Then assign the remaining budget to the lowest cost items. For $q \to 1$, $R_{MIX} \to R_{Q}$, and for $q \to 0$, $R_{MIX} \to R_{C}$
\end{itemize}
% In~\cite{cache-centric-video-recommendation}, cached and related items are placed on the top of the RS list, while the rest of it remains intact. This approach implies an effort towards cache cost minimization in the next request. Moreover,~\cite{chatzieleftheriou2017caching} targets the problem of joint caching and recommendation; if we focus on its recommendation decisions, similarly to~\cite{cache-centric-video-recommendation}, the authors optimize \emph{only} for the next request. In both works, the authors allow some \emph{window} of recommendations for the user satisfaction, and the rest is dedicated to the network gains, which is why these policies could be easily mapped to $q$-Mixed.
In~\cite{cache-centric-video-recommendation}, cached and related items are placed on the top of the recommendation list, while the rest of the list remains intact. Moreover,~\cite{chatzieleftheriou2017caching} targets the problem of joint caching and recommendation.
For some given cache allocation, the RS's objective is to promote items that minimize the caching cost of the next request \emph{only}, ignoring the possibly many subsequent ones. In both works, the authors allow some \emph{window} of recommendations for the user satisfaction, and the rest is dedicated to the network gains, which is why these policies could be effectively mapped to $q$-Mixed.

\vspace{5pt}
\begin{table}[t!]
\centering
\caption{Main Notation}\label{table:notation}
\begin{small}
\begin{tabular}{|l|l|}
\hline
{$\mathcal{K}$}		        		&{Content catalog of size $K$}\\
\hline
{$\lambda$}		            	&{Prob. that the user stays in the session}\\
\hline
{$N$}                       			&{Recommendation batch size}\\
\hline
{$\mathbf{p}_0$}            		&{Baseline popularity of contents}\\
\hline
{$u_{ij}$}					&{Similarity of item $j$ to $i$}\\
%\hline
%{$\mathbf{u}_i$}		    	&{Vector of similarity values for all items related to $i$}\\
\hline
{$U$}		    			&{Adjacency matrix, $U = [\mathbf{u}_1^T, \dots, \mathbf{u}_K^T$]}\\
\hline
{$\mathcal{U}_i(N)$}		    	&{Set of $N$ highest $u_{ij}$ values, related to $i$}\\
\hline
{$\alpha_{ij}$}			    	&{Prob. to click on $j$ when in $i$ from recommendations}\\
\hline
{$w$}			            	&{Recommendation batch, the RS action}\\
\hline
%{$1-a$}					&{Prob. the users requests independently}\\
%\hline
{$r_{ij}$}			        		&{Prob. that $j$ appears in the recommendation batch $w$}\\
%\hline
%{$R$}			            	&{Recommendation policy, i.e., the $\{r_{ij}\}$ values}\\
\hline
{$Q_i(w)$}		                	&{User satisfaction by the recommendation batch $w$}\\
\hline
{$q$}		                		&{Lower level of $Q_i$ enforced by RS}\\
\hline
{$c_i$}			            	&{Network cost of content $i$}\\
\hline
{$S_t$}			            	&{State/Content visited at time $t$}\\
%\hline
%{$\mathcal{C}$}			    	&{Set of cached items with cardinality $|\mathcal{C}| = M$}\\
%\hline
%{useful}					&{Is an item $j$ which is related to $i$ \emph{and} low cost}\\
\hline
% {$\mathbf{c}$}			&{Access cost vector for all contents $i$, of size $K\times 1$}\\
% \hline
%{$\mathcal{C}$}			&{Set of cached content (of cardinality $C$)}\\
%\hline
%{$U_i$}					&{Size of related content list\purple{prosexe to}}\\
\end{tabular}
\end{small}
\end{table}

\section{Problem Formulation and Solution}
\label{sec:mdp-formulation}
We will now cast the problem of optimal sequential recommendations as a Markov Decision Process (MDP) with the objective to minimize the expected cumulative cost in user sessions of arbitrary average length. The user behaviour related to the quality of recommendations will be implicitly taken into account.

\subsection{Defining the MDP}

The MDP is defined by the quadruple $(\mathcal{K}, \mathcal{A}, P, \mathbf{c})$ whose entries refer to the following: as state we consider the currently viewed content, hence the state-space $\mathcal{K}$ is the content catalog. Following the discussion in the previous section about per-item frequencies, the action set $\mathcal{A}$ is the set of all $K\times K$ real matrices $R$, whose entries $r_{ij}\in[0,1]$ determine the frequency of suggesting item $j$ when viewing content $i$. 
Based on the assumptions, the user is Markovian, as her next visited state is fully determined by the current one and not the full history. Moreover, $P$ is the probability transition matrix $K\times K$, where $P_{ij}$ is the probability to jump next to content $j$ if the user currently views content $i$ and essentially serves as the \emph{environment} of the MDP will specify the $P_{ij}$ in the following section. 

We assume that the RS \emph{knows} the user behavior (the $P_{ij}$ dynamics) and optimizes the actions accordingly. 
Learning the user while optimizing (e.g. through reinforcement learning) is deferred to future work. 
%\footnote{We stress that it is not our goal, in this work, to come up with an ultimate/most realistic user model.}
%
Note that we \emph{do not} take into account the time spent on each content by the user nor partial content
%\footnote{Both longer user memory, as well as the impact of partial content viewing are interesting avenues for future extensions.}
viewing, but both variations can be easily integrated in our framework.
Finally, a random item sequence $\{S_{0}, S_{1}, S_{2}, \dots\}$, with $S_{t} \in \mathcal{K}$, corresponds to a random sequence of content costs $\{c(S_{0}), c(S_{1}), c(S_{2}), \dots\}$; hence for some $S_{t} = i$ (the $i$-th content ID), the cost induced to the network is exactly $c_{i}$. 
%Vector $\mathbf{c}$ lists the prices/costs of all the contents $\in \mathcal{K}$.
%The RS assumes a rational behaviour of the user towards the recommendations, that depends on the quality $Q_i(w)$. In the next section will consider two variations of user behavior, but the RS does not learn here the specific user behaviour. It simply tries to offer the best possible recommendations for a reasonably behaving user. 
The following expression gives the transition probability of state evolution in a general way, letting room for further assumptions to be integrated later on in the model
\begin{equation}\label{eq:transition}
P_{ij}=P\{i \to j\} = \alpha_{ij} \cdot r_{ij} + (1 - \sum_{l=1}^K \alpha_{il}\cdot r_{il})\cdot p_{0}(j).
\end{equation}

\begin{comment}
This expression suggests that the user rejects the recommendations with a rate equal to $(1 - \sum_{l=1}^K \alpha_{il}\cdot r_{il})$ which in the case of randomized policies is the mean rejection rate. However, the following calculations show that (\ref{eq:transition}) accurately models the transition from $i \to j$ when $N$-sized batches are recommended.
To do so, we substitute $r_{ij}$ as in (\ref{eq:item-batch}) and also user (\ref{\ref{eq:sumrij}}) as follows,
\begin{align}
& P_{ij} = \alpha_{ij} \sum_{w \in \mathcal{A}_{i}} \mu_{i}(w)\mathbf{1}_{\{j \in w\}} +   \nonumber \\
& (\frac{1}{N} \sum_{j=1}^{K} \sum_{w \in \mathcal{A}_{i}} \mu_{i}(w)\mathbf{1}_{\{j \in w\}} - \sum_{l=1}^K \alpha_{il} \sum_{w \in \mathcal{A}_{i}} \mu_{i}(w)\mathbf{1}_{\{l \in w\}}) p_{0}(j) \nonumber \\
& \sum_{w \in \mathcal{A}_{i}} \mu_{i}(w)   \bigg(   \alpha_{ij} \mathbf{1}_{\{j \in w\}} +  (\frac{1}{N} \sum_{j=1}^{K} \mathbf{1}_{\{j \in w\}} - \sum_{l=1}^K \alpha_{il}  \mathbf{1}_{\{l \in w\}})    \bigg) \nonumber \\
& = \sum_{w \in \mathcal{A}_{i}} \mu_{i}(w) \bigg(   \alpha_{ij} \mathbf{1}_{\{j \in w\}}  + (1 - \sum_{l=1}^K  \alpha_{il} \mathbf{1}_{\{l \in w\}}  ) p_{0}(j)     \bigg) \\
\end{align}
where the last expression uses that $\frac{1}{N} \sum_{j=1}^K \mathbf{1}_{j \in w} = 1$.
\end{comment}

The above expression is somewhat reminiscent of the random web surfer transitions for PageRank \cite{pagerank98}, \cite{MDPpagerank}, and has the following interpretation. 
%Given that the user is at state $i$, the RS suggests content $j$ with frequency $r_{ij}$ (we use the per-item frequencies here, that avoid determining the full batch), the user will click with $\alpha_{ij} \in [0,1]$; this probability depends on the user model and might be a function of the entire recommendation vector $r_{i1}, r_{i2}, \dots, r_{iK}$ (we provide specific expressions in Section \ref{sec:models}).
The user can transit to $j$ in two ways. If content $j$ is in the recommendation batch, the user clicks on $j$ with probability $\alpha_{ij} \in [0,1]$. %(for more specific expressions $\alpha_{ij}$ see Section~\ref{sec:models}).
In the event that the user ignores all of the $N$ items in the batch, she chooses $j$ with probability $p_0(j)$ from personal preferences.
To see why (\ref{eq:transition}) describes exactly this process, we substitute $r_{ij}$ from (\ref{eq:item-batch}) to get,
\begin{align}
& P_{ij} = \alpha_{ij} \sum_{w \in \mathcal{A}_{i}} \mu_{i}(w)\mathbf{1}_{\{j \in w\}} + (\underbrace{ \frac{1}{N} \sum_{j=1}^{K} \sum_{w \in \mathcal{A}_{i}} \mu_{i}(w)\mathbf{1}_{\{j \in w\}}}_{\overset{(\ref{eq:sumrij})}{=} 1} \nonumber \\
&  - \sum_{l=1}^K \alpha_{il} \sum_{w \in \mathcal{A}_{i}} \mu_{i}(w)\mathbf{1}_{\{l \in w\}} ) p_{0}(j) = \nonumber \\
%& \sum_{w \in \mathcal{A}_{i}} \mu_{i}(w)   \bigg(   \alpha_{ij} \mathbf{1}_{\{j \in w\}} +   (\frac{1}{N} \sum_{j=1}^{K} \mathbf{1}_{\{j \in w\}} - \sum_{l=1}^K \alpha_{il}  \mathbf{1}_{\{l \in w\}})    \bigg) \nonumber \\
& \sum_{w \in \mathcal{A}_{i}} \mu_{i}(w) \bigg(   \alpha_{ij} \mathbf{1}_{\{j \in w\}}  + (1 - \sum_{l=1}^K  \alpha_{il} \mathbf{1}_{\{l \in w\}})  p_{0}(j)   \bigg) \nonumber 
\end{align}
% where the last expression uses that $\frac{1}{N} \sum_{j=1}^K \mathbf{1}_{j \in w} = 1$.
Observe that for deterministic policies, there is a single $w$ for which $\mu_i(w) = 1$, the $P_{ij}$ is unique, whereas in the case of randomized policies we view $P_{ij}$ as the average transition probability from $i \to j$, over all batches.
%The probability the user ignores all recommendations when visiting content $i$, is modelled here by the expression $(1-\sum_{l=1}^K \alpha_{il} \cdot r_{il})$, because we take into account the fact that in randomized policies any catalog item can time-share the $N$ available recommendation slots, with frequency $r_{ik}\geq 0$ to be determined by the policy.
% The specific choice of $a_{ij}$ will indeed make this expression a probability in $[0,1]$, e.g. apply the uniform choice $a_{ij}=1/N$ combined with (\ref{eq:sumrij}). In general, $\alpha_{ij}$ could be any function of the policy, i.e., $\alpha_{ij}=f(\mathbf{r}_i, u_{ij})$ depending on how the user is assumed to react. The rejection probability of the recommendation batch is homogeneous for all possible recommended batches (in randomised policies).
% The following Lemma allows for a unique solution to the MDP-based recommendations.
\begin{lemma}\label{lemma:1}
If $\alpha_{ij} < 1$ and $p_{0}(i) > 0~\forall~i,j~\in~\mathcal{K}$, then the MDP ($\mathcal{K}, \mathcal{A}, P, \mathbf{c}$) is unichain, i.e., it has only one class of states for any policy.
\end{lemma}
To prove this, one needs to show that the state-space of the MDP forms an irreducible Markov chain, which is true if all state-pairs are communicating, i.e. $P_{ij} > 0$ in (\ref{eq:transition}). It suffices to consider $p_0(j)>0~\forall~j$ and $\sum_{l=1}^K \alpha_{il} \cdot r_{il}<1$.

\subsection{Optimization Objective}
As explained earlier, we consider a user who consumes sequentially a random number of contents before exiting the session. 
% The session has a random length $L$, where $L$
% follows a Geometric distribution of parameter $\lambda$ and $\overline{L} = \frac{1}{1-\lambda}$.
The induced cost is cumulative over the steps, and since transition probabilities and session-length are random, so is the total cost. The user is considered to start from a given arbitrary state $S_0$ - whose cost is not accounted for as it is outside the recommender's control, and then her session generates a sequence of costs $c(S_{t})$, with $t = 1, \dots, L$, which depends on the policy $R$. Note that the costs in consecutive states are \emph{not} I.I.D., given the Markovian structure of the problem.
%We denote the cost of consuming item $s$ at time instant $t=1,\dots, L$ as $c(s_{t})$. 
The total cost induced by the requests of the user is $\sum_{t = 1}^{L} c(S_{t})$ and our objective is to minimize is the \emph{average} total cost.
% We remind the reader at this point that we consider a user who clicks at an initial file, and then with $\lambda$ he remains in the application and with the complement he exits the session. 
% \theo{maybe here i should mention sth about lambda? or does it suffice to just ref it from the previous section?}
\begin{lemma} \label{lemma:Inf-horizon-mdp}
The average total cost $v(s)$ starting from initial state $S_0=s$ can be written as an infinite horizon cost with discounts
\begin{align}
v(s) &= \mathbb{E}_{s}\left[\sum_{t = 1}^{L} c(S_{t}) \right]= \mathbb{E}_{s}\left[\sum_{t = 1}^{\infty}\lambda^{t-1} \cdot c(S_{t})\right], \label{eq:average-cost1}
\end{align} 
where the $\mathbb{E}_{s}$ stands for conditioning on the starting state being $S_0=s \in \mathcal{K}$, \cite[eq. 4.13]{puterman2014markov}.
\end{lemma}
% As stated in Def.~\ref{def:session}, the probability that the session length $L$ is equal to $l$ is $\mathbf{P}(L = l) = (1 - \lambda) \lambda^{l-1}$ for values of $l~\in~\mathbb{N}^{+}$. 
Equality in (\ref{eq:average-cost1}) holds because the random session length $L$ is assumed to be distributed as a {\tt $Geom(\lambda)$}. The expectation of the total cost is found by applying the law of total expectation
\begin{align}
\mathbb{E}_{s} \bigg(\sum_{t = 1}^{L} c(S_{t})\bigg) = \sum_{l = 1}^{\infty}\mathbb{P}(L = l)\cdot \mathbb{E}_{s}\bigg(\sum_{t = 1}^{l} c(S_{t})\bigg). \label{eq:total-expectation}
\end{align}
We refer the reader to \cite[Prop. 5.3.1]{puterman2014markov} for more details. The parameter $\lambda$ is called the discount factor in the sense that the cost incurred in the immediate future is more important than the cost in the far future. The relative importance of future costs depends on the value of $\lambda \in [0,1]$. 
Starting from state $i \in \mathcal{K}$ we want to minimize
%\begin{tcolorbox}
\begin{problem}[\textbf{Main OP}]
\begin{small}
\label{problem:BASIS}
\begin{align}
v^{\ast}(i) = \underset{R}{\min} & \bigg\{\mathbb{E} \bigg(\sum_{t=1}^{\infty} \lambda^{t-1} c(S_{t}, R)  \mid S_0=i\bigg) \bigg\}~\forall~i~\in~\mathcal{K}\label{eq:optim}\\
\textnormal{subject to} \quad
 & \sum_{j=1}^{K} r_{ij} = N,~\forall i \in \mathcal{K}\label{eq:budget-con}\\
 & 0 \le r_{ij} \le 1,~\forall i, j \in \mathcal{K} \label{eq:probability-con}\\
 & r_{ii} = 0,~\forall i \in \mathcal{K} \label{eq:selfrec-con}\\
 & \sum_{j=1}^{K} r_{ij} \cdot u_{ij} \ge q \cdot Q_{i}^{max},~\forall i \in \mathcal{K} \label{eq:quality-con}.
\end{align}
where $q\in[0,1]$ is the tuning quality parameter.
\end{small}
\end{problem}
%\end{tcolorbox}
$S_t$ is the random variable of the state at step $t$, and $i$ (or $s$) is its realisation taking values in $\mathcal{K}$.
%In what follows, we will use $i$ and $S$ interchangeably to denote the states of the chain.
The optimization variables are the $\left\{r_{ij}\right\}$ per-item recommendation frequencies. The feasible space is shaped by the set of constraints imposed on the RS policy, which has to obey four specifications:
\begin{itemize}[leftmargin=*,noitemsep,topsep=0pt]
    \item (\ref{eq:budget-con}): Recommend exactly $N$ items per content view. 
    \item (\ref{eq:probability-con}): $r_{ij}\in[0,1]$ is a time-sharing proportion.
    \item (\ref{eq:selfrec-con}): No self-recommendation is allowed.
    \item (\ref{eq:quality-con}): Maintain an \emph{average} user satisfaction per viewed content above a pre-defined $q$, i.e.  $\mathbb{E}\left[{Q_i(w)}\right]\geq q$ (see (\ref{eq:user-satisfaction})).
\end{itemize}
%For Eq.(\ref{eq:selfrec-con}), it is enough to restrict the self recommendation frequency $r_{ii} = 0$ for every state.
%eq.(\ref{eq:quality-con}) is a crucial constraint for our problem because it is related to the user satisfaction. Essentially we constrain the average $q_{i}$, see Eq.(\ref{eq:user-satisfaction}), to be above some predefined threshold $Q_{MIN}$. 
Constraint (\ref{eq:budget-con}) incorporates the number of recommendation slots $N$ in the constraints, following (\ref{eq:sumrij}). Using the per-item frequencies the solution complexity becomes insensitive to the value of $N$, something not possible with the initial batch formulation, where the number of batch combinations increases with $N$.

\begin{comment}
Although constraints (\ref{eq:budget-con}), (\ref{eq:probability-con}), (\ref{eq:selfrec-con}) arise naturally, constraint (\ref{eq:quality-con}) seems somewhat forced, as it introduces a hard constraint on the average content quality. This last constraint is another way of integrating assumptions over the user satisfaction in the MDP. We translate it as follows: the RS assumes that the user will be satisfied as long as a content quality of at least $q\cdot Q_i^{max}$ is offered. If $q=1$ the user is satisfied only when recommended the $N$ most relevant content to $i$, whereas recommendation relevance becomes less important as $q$ decreases. When $q=0$ the constraint becomes inactive. Note that, in the three variations of the MDP that follow, only the first model uses this quality constraint. The other two models integrate the user behaviour withing the click-through probabilities $\alpha_{ij}$, which are expressed as a function of the recommendation policy. Such alternatives incorporate user reaction more naturally in the modeling, without introducing arbitrary thresholds like $q$, although eventually such alternatives finally do converge to some average offered average pre-view happiness $\overline{H}$.
\end{comment}

A hard constraint on the average user satisfaction from the recommendation batch is introduced in (\ref{eq:quality-con}). If active, the RS is restricted to maintain an average user satisfaction $\ge q$ for every item $i \in \mathcal{K}$, regardless of how the user reacts to good/bad recommendations. 
%However, as we will see subsequently, such an explicit constraint might not be necessary in user models where the user can immediately assess the quality $u_{ij}$ of the recommended item, and adjust her click probability $\alpha_{ij}$ accordingly.
We denote the feasible set of policies for viewed content $i$ by $\mathcal{R}_i = \{r_{ij}: (\ref{eq:budget-con}), (\ref{eq:probability-con}), (\ref{eq:selfrec-con}), (\ref{eq:quality-con})\}$. The feasible set is denoted by $\mathcal{R}$ and is convex as the intersection of linear inequalities, equalities and box constraints. It is described by $K^2+3K$ linear constraints in total.

\subsection{Optimality}
% A Markov chain is a stochastic process where we observe only a sequence of states such as $\{S_{t}\}_{t>0}$ with $S \in \mathcal{K}$. Furthermore, 
% \theo{I think here we have to make some substantial changes in order to hide general mdp stuff I wrote and reveal some properties for the solution as Tasos commented.}
% Our aim from Eq.(\ref{eq:objective}) is to find the optimal value function per initial state $\mathbf{v}^{\ast} = [{v}^{\ast}(1),\dots, {v}^{\ast}(K)]$. 
% Essentially \textbf{OP~\ref{problem:BASIS}} has optimal substructure and therefore the Bellman Principle is a necessary condition that has to hold.

The optimal solution to \textbf{\nameref{problem:BASIS}}, i.e., the optimal vector $\mathbf{v}^{\ast} = [{v}^{\ast}(1),\dots, {v}^{\ast}(K)]$ for any initial state $s$ is unique and satisfies the Bellman optimality equations (see [Puterman, Theorem 6.2.3] and apply Lemma \ref{lemma:1}).

\myitem{Bellman Optimality Equations.}
%The quest for optimal policies coincides with the quest for finding $\mathbf{v}^{\ast}$. In optimality, $\mathbf{v}^{\ast}$ has to obey the following set of $K$ equations, one per state,
%Finding the optimal value vector, is equivalent to finding the optimal policy. 
The optimal value vector must satisfy the following $K$ (Bellman) equations, one per state,
\begin{align}\label{eq:bellman-opt}
    v^{\ast}(i)  = c_{i} + \lambda \min_{\mathbf{r}_i \in \mathcal{R}_{i}} \bigg\{\sum_{j=1}^{K} P_{ij}(\mathbf{r}_{i}) \cdot v^{\ast}(j) \bigg\}~\forall i \in \mathcal{K}.
% = c(i) + \bar{v} + \lambda \frac{\alpha}{N} \min_{\mathbf{r}_i \in \mathcal{R}_{i}} \bigg\{\sum_{j=1}^{K} r_{ij}\cdot v(j)^{\ast} \bigg\},
\end{align}
where $P_{ij}(\mathbf{r}_{i})$ is defined in (\ref{eq:transition}) and $c_{i}$ indicates the \emph{immediate cost} of visiting state $i$.
%We observe that the minimization problem above involves only the vector $\mathbf{r}_i\in\mathcal{R}_i$, $\forall i$. Furthermore, it is a linear program for $a_{ij}$ constant. In case $a_{ij}=f(\mathbf{r}_i)$ the problem becomes non-linear, depending on the expression.
We can apply well established iterative algorithms to solve these equations~\cite{puterman2014markov,bertsekas1996neuro}; we choose here the well-known policy iteration (PI). A key contribution of our work is that unlike ``vanilla'' MDPs with \emph{discrete actions}, in each iteration we are required to solve a minimization problem in the $K$-sized variable $\mathbf{r}_i$ for each state $i$, see (\ref{eq:bellman-opt}) which radically reduces the interior minimization step of PI. Importantly, as $\mathbf{v}^{\ast}$ remains the same during the policy improvement step (the for loop over the $K$ states), the $K$ minimizers can be straightforwardly parallelized. %For more details on PI we refer the reader to \cite{puterman2014markov}.

\subsection{Versatility of look-ahead horizon via the choice of $\lambda$}

The MDP has the upside of being flexible on the range of problems it can tackle; part of this flexibility is the application of an arbitrary average session length that can be controlled by $\lambda$.
%as well as the possibility to incorporate various assumptions on user response to recommendations, either through the $\alpha_{ij}$ or the constraints. The Policy Iteration algorithm guarantees that all these variations are tractable and can iteratively be solved. In this paragraph we discuss the session-length.
Existing works in the literature analyze \emph{infinitely long} sessions, like in~\cite{giannakas2018show,giannakas2019wiopt}, which is of course unrealistic. The MDP framework introduced in the current work, uses $\lambda$ as a tuning parameter that controls the average session length to be equal to $(1-\lambda)^{-1}$. Let us observe some special cases.

\myitem{Case: $\lambda \to 0$.} 
The objective function in (\ref{eq:optim}), becomes $v(s) = \mathbb{E}_{s} [0^0 c(S_1) + 0^1 c(S_2) + \dots ] = \mathbb{E}_{s} [c(S_1)]$ (with the convention $0^0 := 1$) i.e., the user starting state is $S_0=s$ and does \emph{exactly} one more request which generates loss $c(S_{1})$.
For $\lambda \to 0$, the only future cost is the immediate cost $c_j$ that is incurred by visiting state $S_1=j$ at $t=1$. 
Thus we can explicitly compute $v(s) = \mathbb{E}_{s}[c(S_1)] = \sum_{j=1}^K P_{s,j} c_j$ and find the optimal policy by solving $K$ (one for each starting state $s \in \mathcal{K}$) minimization problems 
\begin{equation}
\underset{\mathbf{r}_s\in\mathcal{R}_{s}}\min \bigg\{\sum_{j=1}^{K} P_{s,j}(\mathbf{r}_i) \cdot c_j \bigg\} \forall~s~\in~\mathcal{K}.
\end{equation}
Setting $\lambda=0$ in \nameref{problem:BASIS}, our MDP returns the $q$-Mixed policy.
%The Bellman equations for $\lambda = 0$ read
%\begin{align}\label{bellman-lambda=0}
%v^{\ast}(i) = c(i) + \cancelto{0}\lambda \underset{\mathbf{r}_i \in \mathcal{R}_{i}}\min \bigg\{\sum_{j=1}^{K} P_{ij}(\mathbf{r}_i) \cdot v^{\ast}(j) \bigg\}
%\end{align}
%%
%Essentially $v^{\ast}(i) = c(i)~\forall ~i \in \mathcal{K}$. Therefore, the optimal one-shot recommendation is found by solving $K$ distinct optimization problems with $v^{\ast}(i) = c(i)$.

%\myitem{Case: $\lambda \to 1/m$ ($m = 2,3,\dots$).}
%This case captures a user who from past statistics, has been measured to have $m$ sequential requests on average after his first request.

\myitem{Case: $\lambda \to 1$.}
For the infinite horizon, the value $v(s)$ diverges for $\lambda = 1$ (no discount) as it allows infinitely many steps to add-up in the cost. However, we can instead find the \emph{time-average} long-term cost (see \cite[Cor.8.2.5]{puterman2014markov}), which is equal to $\lim_{\lambda \to 1} (1-\lambda) v_{\lambda}(s)$ 
%(where $\lambda$ is an index)
. 
This is the limit of the ratio of sum cost over average session-length and it is finite for unichain MDPs (Lemma \ref{lemma:1}).

\myitem{Short and long length $\lambda$.} Since $v(s)$ indicates the expected \emph{cumulative cost-to-go}, for $\lambda \to 0$ the state from which the user starts her session matters, and the values of the vector $v(s)$ will differ. On the contrary for $\lambda\to 1$, as the $v(s)$ tend to infinity ($v(s)$ is cumulative and larger as $\lambda$ grows), the relative difference between the states becomes negligible. That is all the states have approximately the same value and the starting state $s$ does not matter.
%\theo{Taso, auto to xes se ena simeio pou de mou fainotan appropriate, to metakinisa edw}

In reality, the session length is somewhere in the middle. To determine the value of $\lambda$ in practice, we use the fact that the mean session-length is equal to $1/(1-\lambda)$. Then, the RS could measure empirical averages of the user session lengths to determine an estimate $\hat{\lambda}$ and substitute this value in the MDP to derive appropriate recommendations.
% From complexity standpoint, as $\lambda \to 1$, the algorithm becomes slower in convergence (see convergence rate of VI). Nonetheless, the extra computational time offers a much deeper vision of rewards. 
% Recall that larger $\lambda$ means that we value the next time instance as (almost) equally important as the current one. 
% If we see that recursively, it means that given we are at time instance $t$, all future values are as important as $t$.
% It becomes evident that since we are dealing with an infinite horizon problem, as $\lambda$ grows we essentially solve the $(1-\lambda$)-optimal of the paper in~\cite{giannakas2018show}. 
% Essentially, the decision-maker (or agent) is not simply trying to figure out ways to save cost from the imminent requests, but it tries to figure out the contents which will lead to \say{long winning paths}.
% On the other hand, when $\lambda \to 0$, having a long session is considered very unlikely, thus the agent aims for the short-sighted optimal solution and has no interest whatsoever on future savings.

%% Section 4
\section{A User Model}
\label{sec:models}
The transition probabilities $P_{ij}$ in (\ref{eq:transition}) allow to model various types of user response to the recommendation policy. The user behaviour is summarised by her click-through probabilities $\left\{a_{i,j}\right\}$. These can take specific expressions and be functions of $u_{ij}$ or even $r_{ij}$ to represent some \say{reactivity} of the user to the policy, and each choice represents a different type of behaviour. Note that as long as $P_{ij}$ is a convex or linear function of $r_{ij}$, our framework can solve it optimally and efficiently.
%
%\myitem{A ``curious'' user:}
We study here a specific user, who remains equally curious about any recommended item, provided that the long-term quality of recommendations remains reasonably good. Then, the response of this \textbf{``curious'' user} is:
%We attempt to capture such a user with the following assumptions: 
\begin{equation}
\alpha_{ij} = \alpha/N \ \ \ \ \text{for}\ \ \ \ \mathbb{E}[Q_i] \ge q~\forall i.
\end{equation}
The above reads that the user wants to be satisfied by at least $q$ from the RS recommendations (see (\ref{eq:quality-con})). Her transition probabilities (i.e., (\ref{eq:transition})) now become
\begin{eqnarray}
    \label{eq:transition-model1}
    P\{i \to j\} & = & \frac{\alpha}{N} \cdot r_{ij} + (1 - \alpha) \cdot p_0(j),
\end{eqnarray}
where notice that $\sum_{l=1}^K \frac{\alpha}{N} \cdot r_{ij} = \alpha$. If the user is satisfied, her \emph{expected} click-through rate $\alpha$ remains fixed throughout the session, and she may click any of the $N$ recommended items uniformly at random. 
%This type of user is reminiscent of the PageRank web surfer, where the user clicks any hyperlink with a fixed probability, or jumps to a random page~\cite{pagerank98}.
%
%For example, the quality constraint might capture the user's long-term trust in the RS (i.e., a time-scale beyond a single session). If this is violated the user leaves the service altogether. Yet, given such trust, the user is curious enough about any recommended item.
Another, perhaps more pragmatic interpretation, is that the RS can just measure some data regarding the user's clickthrough rate and how this relates to the average user satisfaction, and uses these estimates to calibrates the MDP. The tuple $(\alpha, q)$ comprises a wide range of user attitudes ranging from highly curious (high $\alpha$, low $q$) to rather conservative (medium/high $\alpha$, very high $q$).

\myitem{MDP Solution.}
%The average rejection rate, with which the user ignores the recommendations is the probability $1-\alpha \in [0,1]$, and consequently t
%The total click-through probability, with which the user follows the RS, is just $\alpha$. 
%This is the simplest way that the RS can model user interaction, assuming click-through probability $(1-\alpha)$, session-length probability $\lambda$ and recommendation threshold $q$, as a-priori user defined and somewhat independent of the recommendation quality. The RS here just assumes these as known or measured, and takes these into account when designing the policy. This model is a first step to take user behaviour into account and offer recommendations above some quality, while avoiding to include a truly adaptive user with content preferences.
Using (\ref{eq:transition-model1}), the Bellman equations take the form, $\forall i \in \mathcal{K}$,
\begin{align}\label{eq:bellman-opt-model1}
v^{\ast}(i)  = c_{i} + \bar{v} + \lambda \frac{\alpha}{N} \min_{\mathbf{r}_i \in \mathcal{R}_{i}} \bigg\{\sum_{j=1}^{K} r_{ij}\cdot v^{\ast}(j) \bigg\},
\end{align}
where $\bar{v} := \lambda(1-\alpha)\sum_{j=1}^{K} p_{0}(j) v^{\ast}(j)$ is independent of $i$.
Therefore, in each greedy improvement step, the optimizer will have to solve the following optimization problem.
%\begin{tcolorbox}
\begin{problem}[\textbf{Inner OP}]
\begin{small}
\label{problem:flat}
\begin{align}
\underset{\mathbf{r}_{i}\in \mathcal{R}_{i}}{\min} & \bigg\{\sum_{j=1}^{K} r_{ij}\cdot v^{\ast}(j) \bigg\} \nonumber \\
\textnormal{subject to} \quad
 & \mathcal{R}_i = \{r_{ij}:(\ref{eq:budget-con}), (\ref{eq:probability-con}), (\ref{eq:selfrec-con}), (\ref{eq:quality-con})\} \nonumber
\end{align}
where $q\in[0,1]$ is the tuning quality parameter of (\ref{eq:quality-con}).
\end{small}
\end{problem}
%\end{tcolorbox}
\begin{lemma}
The greedy improvement step of Policy Iteration for the curious user, reduces to solving the \nameref{problem:flat} which is a Linear Program (LP) of size $K$; the objective and all the constraints are linear on the variables $r_{ij}$.
% minimization step 4: Algorithm~\ref{alg:vi} (see Eq. (\ref{eq:bellman-opt-model1})) is a linear program (LP).
\end{lemma}
%\begin{proof}
%Observe in Eq. (\ref{eq:bellman-opt-model1}) that the objective in the minimization is linear in the variable $r_{ij}$. Moreover, the solution space is convex (we remind the reader that the feasible solution set is $\mathcal{R}_i = \{r_{ij}:(\ref{eq:budget-con}),(\ref{eq:probability-con}),(\ref{eq:selfrec-con}),(\ref{eq:quality-con})\}$). 
%\end{proof}
The $K$ LPs in the inner loop of PI can be solved using standard software (e.g. CPLEX). Note here, that solving the MDP returns a randomised policy in general, due to the constraint (\ref{eq:quality-con}). Moreover, the Bellman equations reveal structural properties of the policy, showing optimality for myopic heuristics as special cases.
\begin{property}\label{prop:q=1}
For $q = 1$, the optimal policy is the Top-N.
\end{property}
\begin{proof}
For $q=1$, the rhs of (\ref{eq:quality-con}) becomes $1\cdot \sum_{l \in \mathcal{U}_{i}(N)} u_{il} = Q_i^{\max}$. Assume that the optimal policy for content $i$ is to assign $r_{ij} = 1$ to contents that correspond to $\mathcal{U}_{i}(N-1)$, and $r_{ij} = x > 0$ to some content with $u_{ij} \notin \mathcal{U}_i(N)$ and $r_{im} = 1-x$ to the least related item with $u_{im} \in \mathcal{U}_{i}(N)$.
Then the constraint (\ref{eq:quality-con}) reads
\begin{align}
    \underset{l \in \mathcal{U}_{i}(N-1)}{\sum u_{il}} + (1-x) u_{im} + x u_{ij} & \ge \underset{l \in \mathcal{U}_{i}(N-1)}{\sum u_{il}} + u_{im} \nonumber \\
    (u_{ij} - u_{im}) \cdot x & \ge 0
\end{align}
By definition, $u_{im} > u_{ij}$ and thus the inequality cannot hold if we assign a positive budget to any item $j$ with $u_{ij} \notin \mathcal{U}_{i}(N)$.
\end{proof}
%In the case where $u_{ij} \in \{0,1\}$ the optimal policies can be more than one.
%\theo{this one is general i think}
%\begin{property}\label{prop:same-solutions}
%If two contents $i$ and $j$ have $\mathbf{u}_{ij} = \mathbf{u}_j$, then $\mathbf{r}_{i}^{\ast} = \mathbf{r}_{j}^{\ast}$.
%\end{property}
%\begin{proof}
%In (\ref{eq:bellman-opt-model1}), the objective of the minimization is the same regardless of the state and is $\min_{\mathbf{r}_i \in \mathcal{R}_{i}} \bigg\{\sum_{j=1}^{K} r_{ij}\cdot v^{\ast}(j) \bigg\}$. (\ref{eq:budget-con}), (\ref{eq:probability-con}) hold for $i$ and $j$. However, if $\mathbf{u}_i = \mathbf{u}_j$, then $u_{ii} = 0 \to u_{ij} = u_{ji} = u_{jj} = 0$ which implies that $r_{ij}$ However, the vector $\mathbf{v}^{\ast}$ is the same for all states. Thus all states with the same solution space will have the same policy.
%\end{proof}
\begin{property}\label{prop:q=0}
For $q=0$ the optimal policy is the Low Cost.
\end{property}
\begin{proof}
Assume that we can order the optimal values $v^{\ast}(i)$ in increasing order $v^{\ast}(1) < \dots < v^{\ast}(K)$.
To find $v^{\ast}(i)$ we need to solve $\min_{\mathbf{r}_i \in \mathcal{R}_{i}} \bigg\{\sum_{j=1}^{K} r_{ij}\cdot v^{\ast}(j) \bigg\}$. 
For the case $q=0$, we can analytically compute $v^{\ast}(i)$, because the optimal decision is to assign $r_{ij} = 1$ to the lowest $v^{\ast}(j)$ (excluding $v^{\ast}(i)$). We can identify two cases for the expression of $v^{\ast}(i)$. Case (a): If $1 \le i \le N$ then 
\begin{align}\label{v1}
v^{\ast}(i) = c_i + \bar{v} + \lambda \left(\sum_{j = 1: j \neq i}^N v^{\ast}(j) + v^{\ast}(N+1) \right)
\end{align}
where in the above expression we need to make sure we exclude the self recommendation from the evaluation. Else for Case (b): $i > N$, the expression becomes
\begin{align}\label{v2}
v^{\ast}(i) = c_{i} + \bar{v} + \lambda \sum_{j = 1}^N v^{\ast}(j).
\end{align}
We need to compare the values of the states in pairs. There are three possibilities for the pairs. Pair-case (I): $1 \le i,j \le N$ and $i < j$, we get the difference (using (\ref{v1}))
\begin{align}
v^{\ast}(i) - v^{\ast}(j) < 0 \Rightarrow (c_i - c_j) + \lambda (v^{\ast}(j) - v^{\ast}(i)) < 0,\nonumber
\end{align}
where for the second term above $N-1$ terms have cancelled out.
Notice that due to the ordering, $\lambda (v^{\ast}(j) - v^{\ast}(i)) > 0$, so it must hold that $c_i - c_j < 0$ for the above expression to have a negative sign. Pair-case (II): $1 \le i \le N$ and $j > N$, we use  (\ref{v1}) and (\ref{v2}) and we result in the exact same expression for their difference as above. Finally, for Pair-case (III): $N < i < j$, we use (\ref{v2})
\begin{align}
v^{\ast}(i) - v^{\ast}(j) &< 0 \Rightarrow c_i - c_j < 0.
\end{align}
%which immediately states that if $v^{\ast}(i) - v^{\ast}(j) < 0$ then $c(i) - c(j) < 0$.
Therefore, the optimal costs-to-go $v^{\ast}(i)$ are ordered exactly as the immediate costs $c_i$, which concludes that for content $i$, recommending the $N$ lowest costs excluding $i$ is optimal.
\end{proof}

\section{Validation}
\label{sec:simulations}
We evaluate the performance of the proposed MDP recommendation policy in terms of cost and user satisfaction (against other myopic ones) and of computational efficiency.

% and compare it with other myopic approaches. We illustrate its optimal performance and gain a better understanding over the structure of the optimal policies.

% in doing so, we will try to highlight key conclusions in order to gain a better understanding of the NFR problem.
% For realistic evaluation, we will use two video/audio datasets, as well as a number of synthetic traces.

\subsection{Simulation Setup}
\myitem{Caching Policy and Baseline Cost.}
In our simulations, we assume that for each dataset, the number of cached items is $M = 0.01 \cdot K$, where $K$ is the size of the corresponding catalog. This number is similar to other works~\cite{femto} or~\cite{chatzieleftheriou2019TMC}. 
We cache the first $M$ IDs of the catalog, i.e., $\mathcal{C} = \{1,\dots, M\}$. 
We consider a uniform personal preference distribution $\mathbf{p}_0$ i.e., $\mathbf{p}_{0}\sim Uniform(1, \dots, K)$. Thus, the caching policy is essentially random and the performance of all policies will be unaffected by $\mathbf{p}_{0}$.
We proceed like that as our goal is to understand the \emph{true gains} that come from the RS's \emph{network friendliness}, and \emph{not} the ones hailing from the potential skewness of $\mathbf{p}_0$. %Therefore, the uniform popularity and the random caching policy set the stage for the RS and help us better understand the \emph{true gains} that come only from the RS.
% that our recommendation policy results do not depend on some \say{more favorable} caching policy, such as most popular etc.
Furthermore, the cost of the \emph{non}-cached items is set to an arbitrary price, say 10.0 units, and of the cached ones to an arbitrary smaller price, say 0.0. If we assume that \emph{there is no RS in place}, or equivalently the user never follows recommendations and all requests are generated according to $\mathbf{p}_0$ (i.e., the standard Independent Reference Model), the hit probability $p_{hit} = 0.01$, which easily translates to $\overline{C} = 10.0 \times (1-p_{hit}) + 0.0 \times p_{hit} = 9.9$ units of cost per request. The numbers above hold for all the plots of the section since uniform $\mathbf{p}_0$ and $M/K=0.01$ hold everywhere.

% the situation where the RS is aware of the we can see, the uniform selection of $\mathbf{p}_{0}$ can help us distinguish the \emph{true gains} we receive from the \emph{network friendliness} of the RS and \emph{not} from the popularity skewness

%Importantly, in doing so the effect of the RS on the average cost becomes clearer as the  cache hits are attracted through the recommendation mechanism and not the sheer power of some content through its high $p_{0}(i)$ mass.
% Moreover, in the captions of the figures, we will note the $N$ and the assumed average horizon (session length) $\overline{L}=1/(1-\lambda)$.

\myitem{Simulation and Metrics.}
%When we refer to the \say{Flexible} or the \say{Peaky} user, we implement a user who behaves accordingly. 
The RS \emph{knows} exactly the user model and the evaluated policy is the optimal one computed by the MDP.
%In the case of the \say{Flexible} user, the optimal RS policy is probabilistic, and in order to implement it and suggest the $N$ contents we follow the strategy of~\cite{blaszczyszyn2015optimal}.
The first metric is the average cost, denoted as $\overline{C}$ (which is network-oriented) and the second one is the average user satisfaction denoted as $\overline{Q}$; both are measured \emph{per content request}.
%We then focus on two additional metrics which are user-satisfaction related: the first one being the expected pre-view happiness $\overline{H}$ (see Eq. (\ref{eq:user-satisfaction}) and the second being the post-view happiness she receives from the clicked recommendations, which we denote as $\overline{u^{rec}}$.
We perform a Monte Carlo simulation where we generate 1000 sessions of random size $L$, where $L \sim Geom(\lambda)$ ($\lambda$ is a parameter of the simulation). Therefore, we measure $C_{L}$ (see (\ref{eq:objective})) and $Q_{L}$ (see (\ref{eq:user-satisfaction}))
%\begin{equation}\label{metrics}
%C_{L} = \frac{1}{L}\sum_{t=1}^{L} c(S_t)~~~\text{and}~~~Q_{L} = \frac{1}{L}\sum_{t=1}^{L} Q_{S_{t}}(w) \\
%\end{equation}
%where $C_{L}$ is defined in Eq. (\ref{eq:objective}) and is the average cost of the session, $Q_{L}$ the corresponding average user satisfaction, and $Q_{S_{t}}(w)$ is defined in Eq. (\ref{eq:user-satisfaction}).
%More specifically, $w_{n}$ denotes the batch that appeared in time instance $n$. 
For some fixed $\lambda$, the quantities $\overline{C}$ and $\overline{Q}$ are produced by further averaging $C_{L}$ and $Q_{L}$ over 1000 runs. 
In some experiments we will also report the cache hit relative gain of MDP against other policies which will be defined as $\frac{p_{hit}^{MDP}-p_{hit}^{ref}}{p_{hit}^{ref}}\cdot 100\%$.

\myitem{Execution of the PI Algorithm.}
All experiments were carried out using a PC with RAM: 8 GB 1600 MHz DDR3 and Processor: 1,6 GHz Dual-Core Intel Core i5. The minimizers of \nameref{problem:flat} that arise were solved through CPLEX.

\subsection{Traces}
% \myitem{lastfm.} ($K = 757$)
% We consider a dataset from the last.fm database \cite{lastfm}. We applied the ``getSimilar'' method to the content IDs' to fill the entries of the matrix $\mathbf{U}$ with similarity scores in [0,1]. Then, we (a): find the largest component of the graph, then (b): set scores above $0.1$ to $u_{ij} = 1$ to obtain a dense binary $\mathbf{U}$ matrix and finally (c): remove rows and columns with less than three related items.
%To test the policies discussed throughout the previous sections, we need some content libraries and most importantly the respective content relation graph $U$ (Section \ref{subsec:inputs}).
%a content relation graph $\mathcal{U}$. 
We use three datasets to construct three content relation graphs $U$ (Section \ref{subsec:inputs}), two real ones and one synthetic.

\myitem{MovieLens.} We consider the Movielens movie rating dataset~\cite{movielens}, containing $69162$ ratings (0 to 5 stars) of $671$ users for $9066$ movies. We apply an item-to-item collaborative filtering 
% (using 10 highest cosine similarity) 
to extract the missing user ratings, and then use the cosine similarity with range $[-1,1]$ of each pair of contents. We floor all values $< 0.5$ to zero and leave the remaining ones intact.
%assume there is a connection on the graph if the cosine distance is greater than 0.5. 
Finally, we remove from the library contents with less than 25 related items to end up with a relatively dense $U$.

\myitem{YouTube.}
We consider the YouTube dataset found in~\cite{youtube-related-videos-dataset}. 
% We are mainly interested in acquiring the list of related items (i.e., the true recommendation list as YouTube would present it) of each item. 
From this, we choose the largest component of the library and build a graph of 2098 nodes (contents) if there is a link from $i \to j$. As the values of the dataset were $u_{ij} \in \{0,1\}$, whenever an edge was found, we assigned it a random weight $u_{ij} \sim Uniform(0.5,1)$.

% offered, (b) helps to achieve a more realistic structure of graph $G$. 
% \myitem{Note:} For both datasets, and for every connection found in the graph from $i \to j$ we draw a rand(0.5,1) and assign this value as $u_{ij}$.

\myitem{Synthetic.}
We consider a synthetic content graph $U$; this way we can see how the algorithm behaves in a more uniform $U$. We decide the size of the library $K$, and for every item in the library we draw a number out of $Uniform(1,100)$ which serves as the number of neighbors of $i$. We then assign on the edges a random weight $u_{ij} \sim Uniform(0.5,1)$.

For these datasets, we present the statistics related to $U$, and its relation to the cached contents. To this end, based on $U$, we consider there is an edge from $i \to j$ if $u_{ij}>0$ and we are interested on the out-degree of the nodes. The graph in general is directed.
%we define the set of friends $\mathcal{F}_i$ as the set contents that are connected to $i$, with $|\mathcal{F}_i|$ being the number of friends/neighbors of $i$. The main reason we do that is to help us better understand performance differences \emph{between} the datasets.
\begin{itemize}
    \item $\text{deg}^{-}_{i}$: out-degree of node $i$.
    \item $\text{deg}^{-}_{i}(\mathcal{C})$: out-degree of node $i$ directed \emph{only} to nodes in the set $\mathcal{C}$ (the set of cached items).
    %\item $D_{\mathcal{C}}(i)$: we find the shortest paths of every node $i \in \mathcal{C}$ to all the other nodes which are also $\in \mathcal{C}$. This makes up for $\frac{M \cdot (M-1)}{2}$ shortest path distances over which we average.
\end{itemize}

% \begin{figure}
% \centering
% \includegraphics[width=0.5\columnwidth]{histograms_loglinear}\label{fig:ccdf}
% \caption{Traces: Node Out-degree CCDF}
% \label{fig:traces-all}
% \end{figure}

% \begin{figure}
% \centering
% \subfigure[Traces: Node Out-degree, $\text{deg}^{-}_{i}$ CCDF]{\includegraphics[width=0.42\columnwidth]{histograms_loglinear}\label{fig:ccdf}}
% \hspace{0.04\columnwidth}
% \subfigure[Brute Force vs Proposed]{\includegraphics[width=0.45\columnwidth]{vanilla1}\label{fig:vanilla}}
% \caption{Traces and Brute Force MDP}
% \label{tr-and-br}
% \end{figure}

\begin{figure}
\centering
\includegraphics[width=0.44\columnwidth]{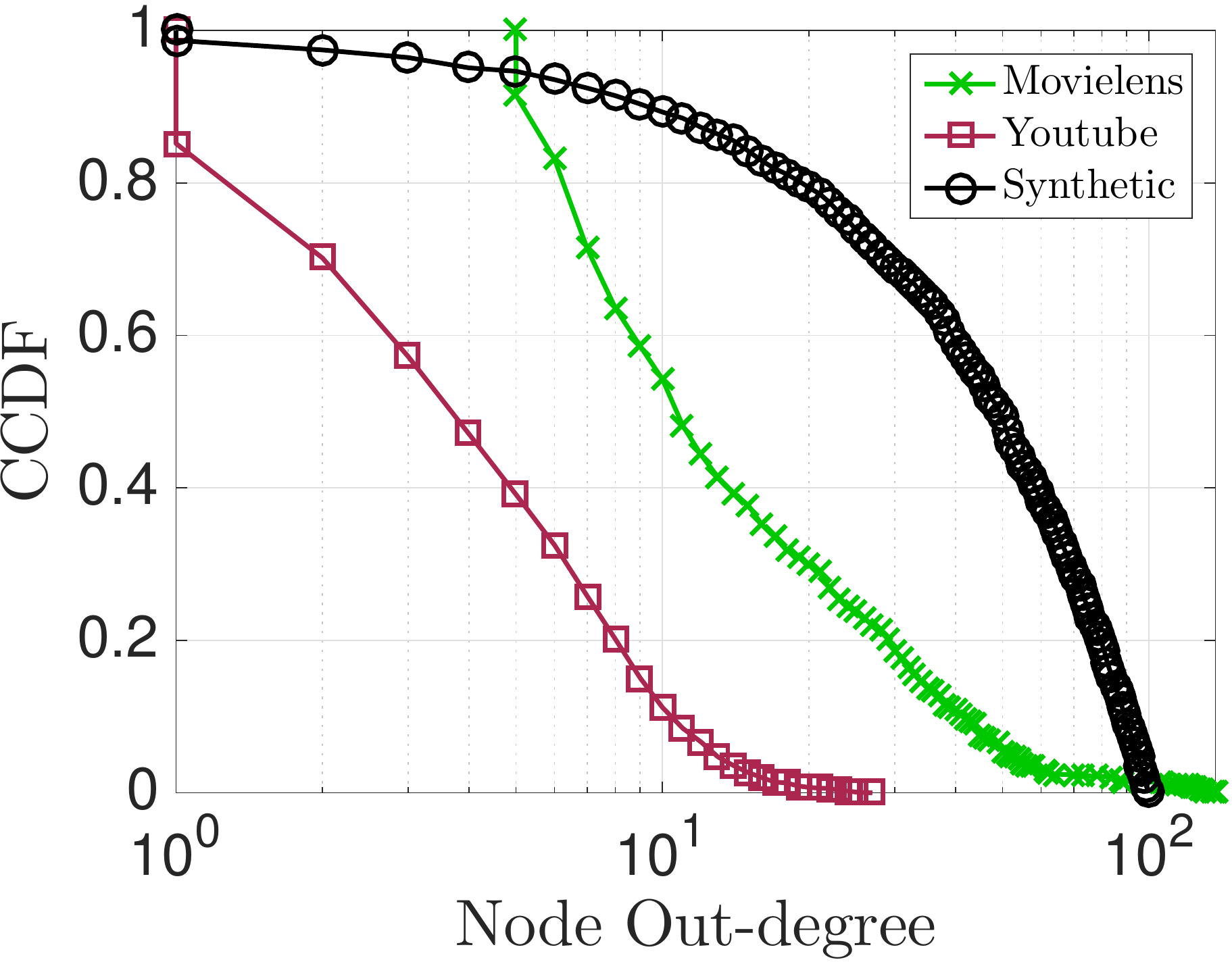}
% \subfigure[Brute Force vs Proposed]{\includegraphics[width=0.45\columnwidth]{vanilla1}\label{fig:vanilla}}
\caption{Traces Out-Degree CCDF}
\label{fig:ccdf}
\end{figure}

In Fig.\ref{fig:ccdf}, on the $x$-axis we see the $\text{deg}^{-}_{i}$ in logarithmic scale, and on the $y$-axis its ccdf. We can conclude for the two real traces, that the $\text{deg}^{-}_{i}$ tends to be quite high only for a small fraction of the nodes. 
%This metric will be useful to interpret the optimal policy's actions later. 

%Finally, note that for the remaining of this section we will frequently use the term \say{useful content} for recommendation when we are at state $i$, this means that
%In Fig. \ref{fig:traces-all}, on the $x$-axis we see the content ids, and on the $y$-axis the number of neighbors this item has. All three plots are not in the same, as each trace is different and by setting the $y$-axis tight, we can better observe the pmf structure (skewed or not etc). It is evident that the real traces experience higher variance and the $\mathcal{F}_i$ tends to be more skewed, in contrast to the synthetic one which was chosen as a Poisson graph.

%\begin{table*}[t!]
%\centering
%\caption{Statistics - Datasets} \label{table:stats}
%\begin{tabular}{c|c|c|c|c|c|c|c}
%				&{Nodes} 		&{Total Edges} 		&{mean $\text{deg}^{-}$} 		&{$std(\text{deg}^{-})$}  		&{mean $\text{deg}^{-}(\mathcal{C})$} 		&{$std(\text{deg}^{-}(\mathcal{C}))$} 		&{$\overline{D_{\mathcal{C}}}$} \\
%\hline \hline
%{MovieLens}			    		&{1060} &{20162} &{19.02}&{19.61} &{0.17}&{0.20} &{2.22}\\
%\hline
%{YouTube}		        			&{2098} &{11288} &{5.38}&{4.16} &{0.06}&{0.05} &{10.11}\\
%% \hline
%% {last.fm}			    		&{7.87} &{30.70} &{0.0845} &{0.0827} \\
%\hline
%{Synthetic}			    		&{2000} &{99367} &{49.68}&{28.84} &{0.51}&{0.56} &{1.51}\\
%\end{tabular}
%\end{table*}

\begin{table}
\centering
\caption{Statistics - Datasets} \label{table:stats}
\begin{tabular}{c|c|c|c}
							&{MovieLens} 		&{YouTube} 		&{Synthetic} \\
\hline \hline
{Nodes}			    			&{1060} 			&{2098} 			&{2000}\\
\hline
{Total Edges}		        			&{20162} 			&{11288} 			&{99367}\\
\hline
{mean $\text{deg}^{-}$}			&{19.02} 			&{5.38} 			&{49.68}\\
%\hline
%{$std(\text{deg}^{-})$}			&{19.61} 			&{4.16} 			&{28.84}\\
\hline
{mean $\text{deg}^{-}(\mathcal{C})$} &{0.17} 			&{0.06} 			&{0.51}\\
%\hline
%{$std(\text{deg}^{-}(\mathcal{C}))$}	&{0.20} 			&{0.05} 			&{0.56}\\
\hline
%{mean $D_{\mathcal{C}}$}			&{2.22} 			&{10.11} 			&{1.51}\\
\end{tabular}
\end{table}

\begin{figure}
\centering
\subfigure[MovieLens - Cost]{\includegraphics[width=0.4\columnwidth]{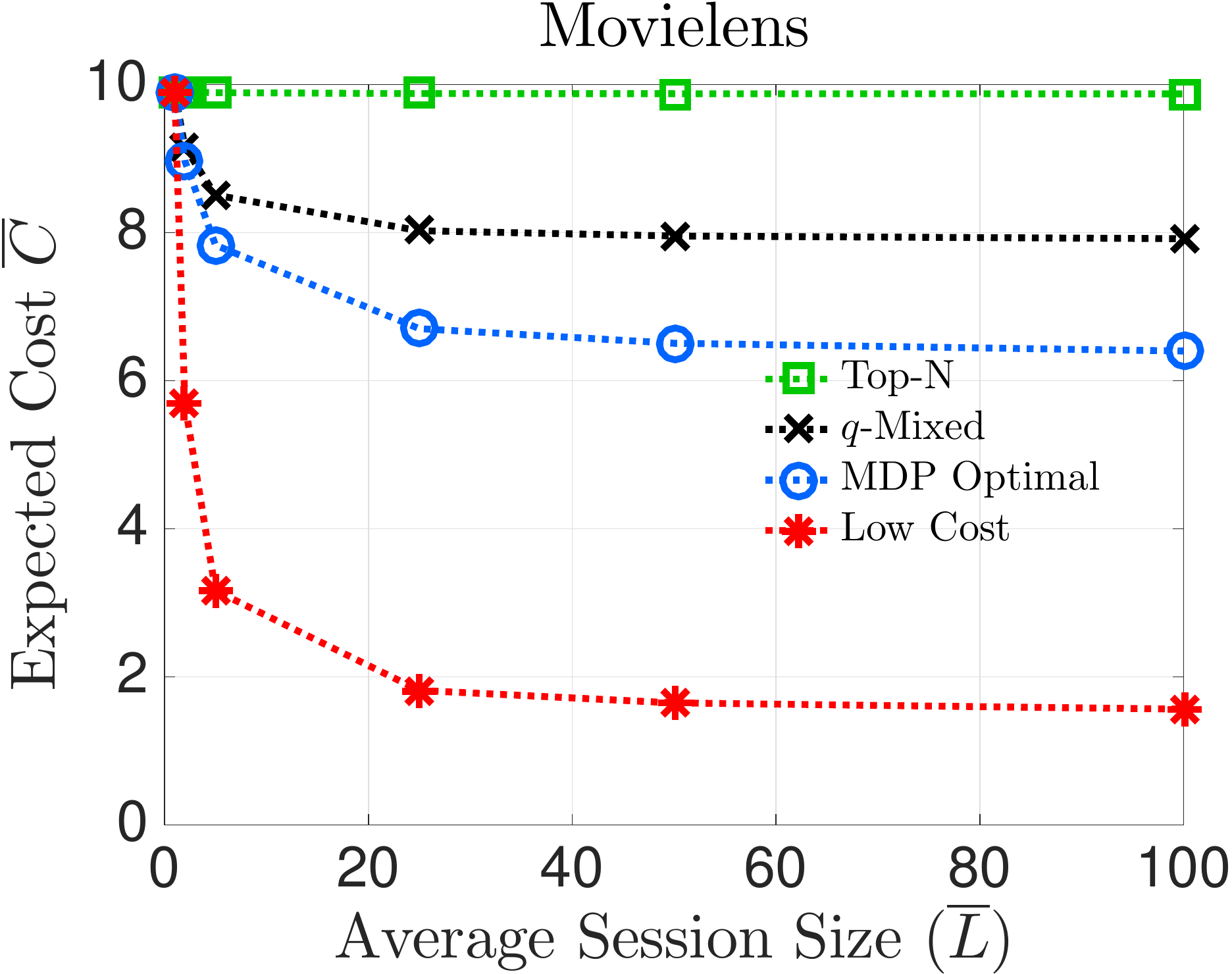}\label{fig:cost-movielens}}
\hspace{0.04\columnwidth}
\subfigure[YouTube - Cost]{\includegraphics[width=0.4\columnwidth]{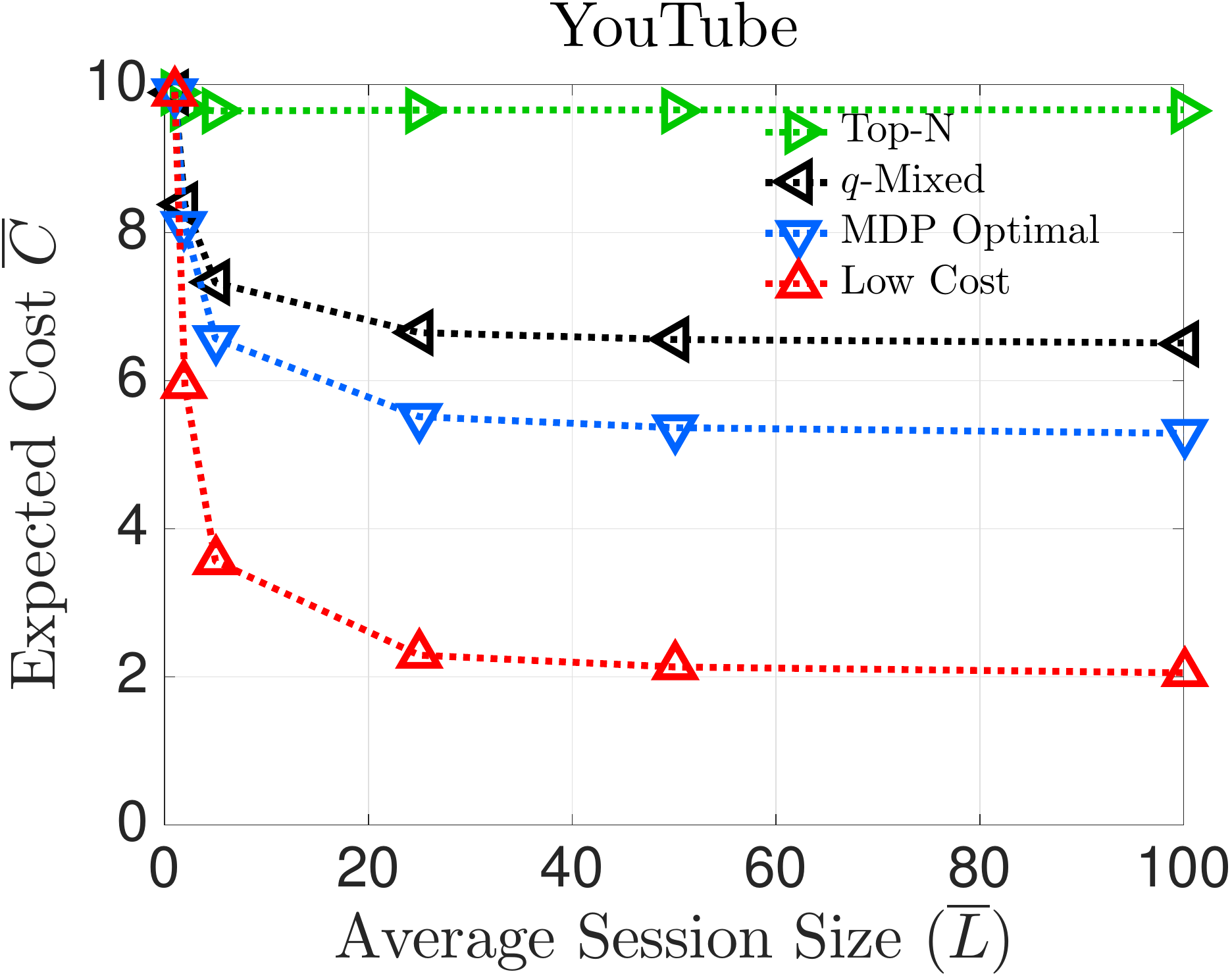}\label{fig:cost-youtube}}
%%%%%%%%%%%%%%%%%%%%%%%%%%
%%%%%%%%%%%%%%%%%%%%%%%%%%
%%%%%%%%%%%%%%%%%%%%%%%%%%
%%%%%%%%%%%%%%%%%%%%%%%%%%
\hspace{0.04\columnwidth}
\subfigure[Movielens - User Satisfaction]{\includegraphics[width=0.4\columnwidth]{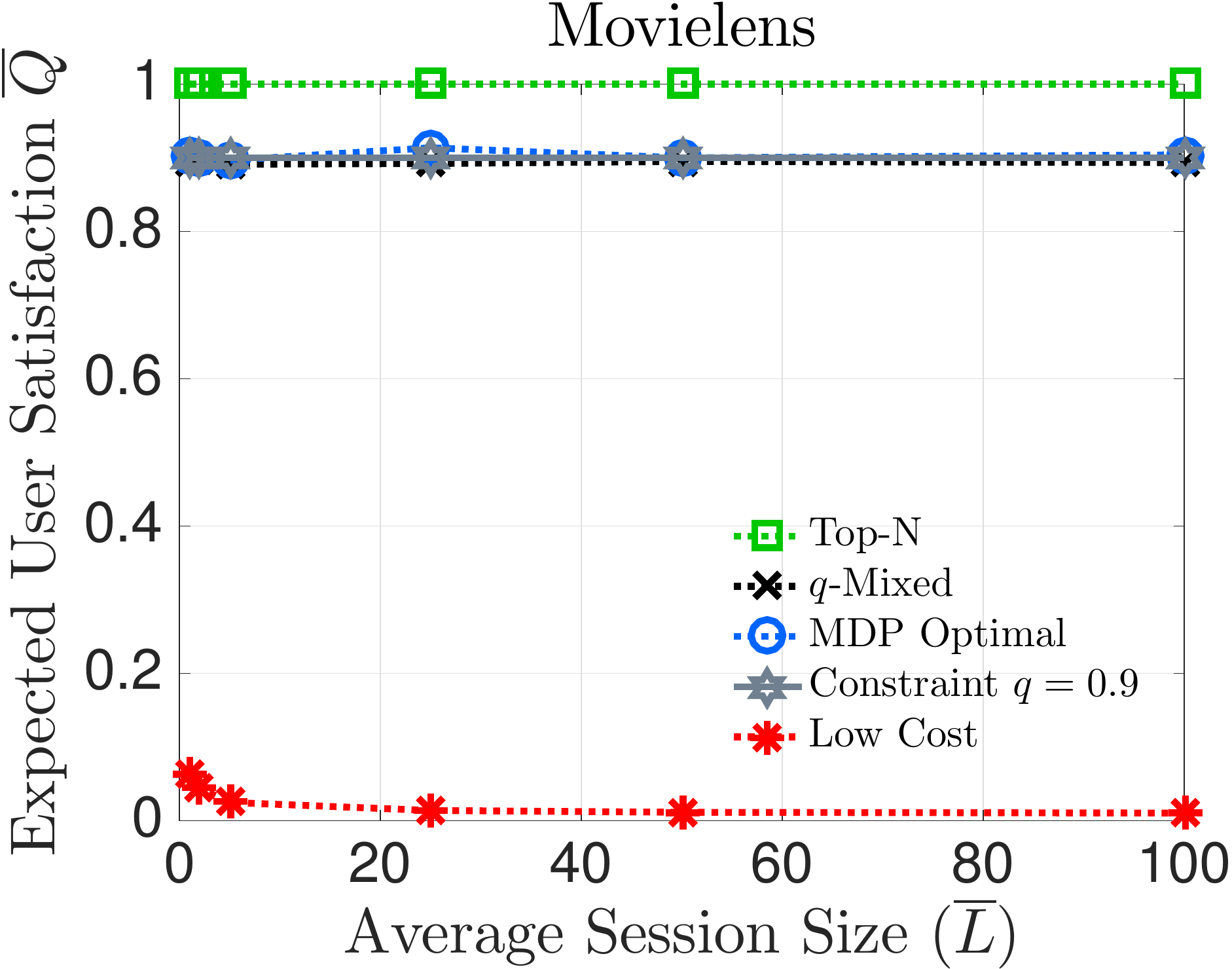}\label{fig:satisfaction-movielens}}
\hspace{0.04\columnwidth}
\subfigure[YouTube - User Satisfaction]{\includegraphics[width=0.4\columnwidth]{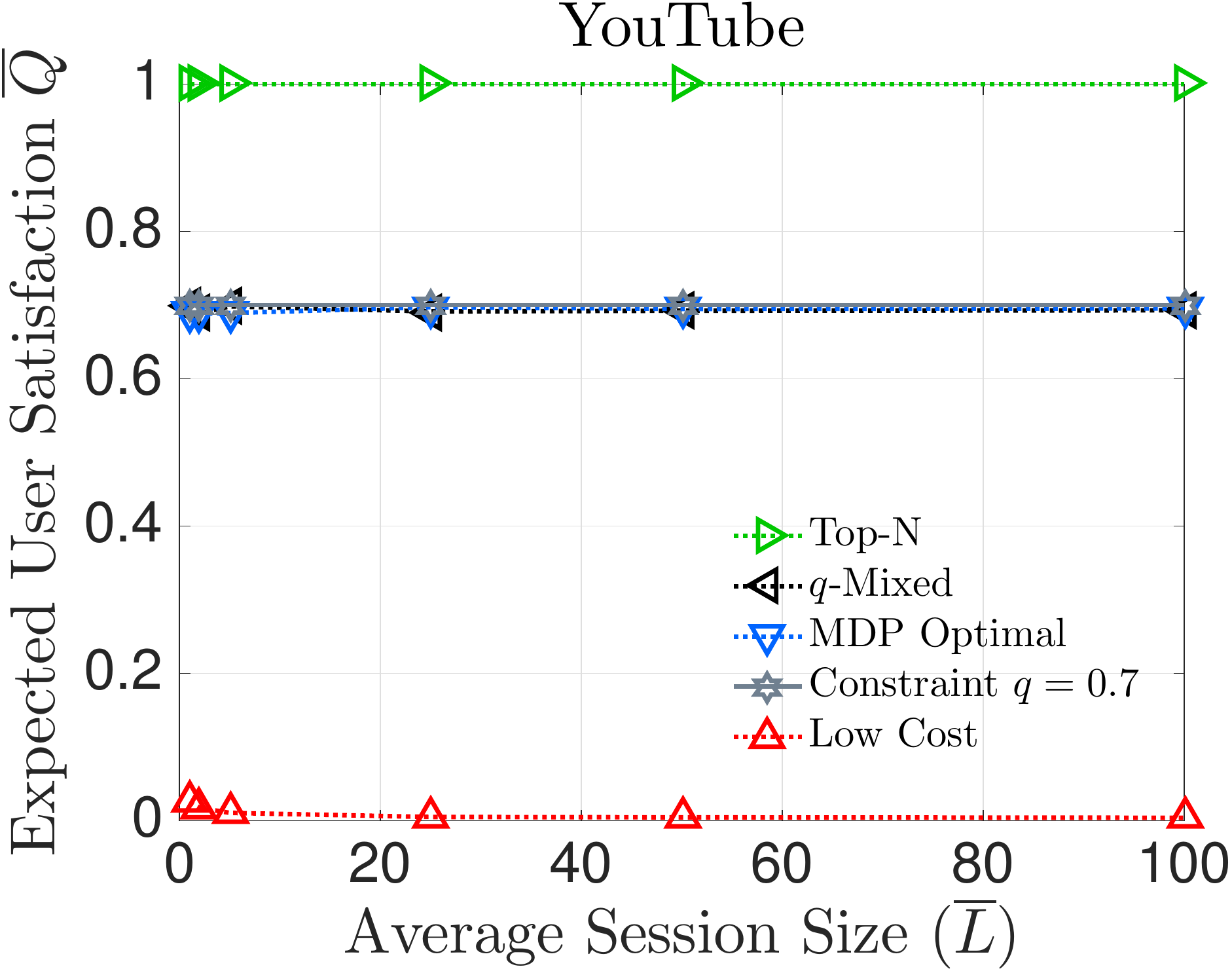}\label{fig:satisfaction-youtube}}
%%%%%%%%%%%%%%%%%%%%%%%%%%
%%%%%%%%%%%%%%%%%%%%%%%%%%
%%%%%%%%%%%%%%%%%%%%%%%%%%
%%%%%%%%%%%%%%%%%%%%%%%%%%
\hspace{0.04\columnwidth}
\subfigure[Relative Gain]{\includegraphics[width=0.4\columnwidth]{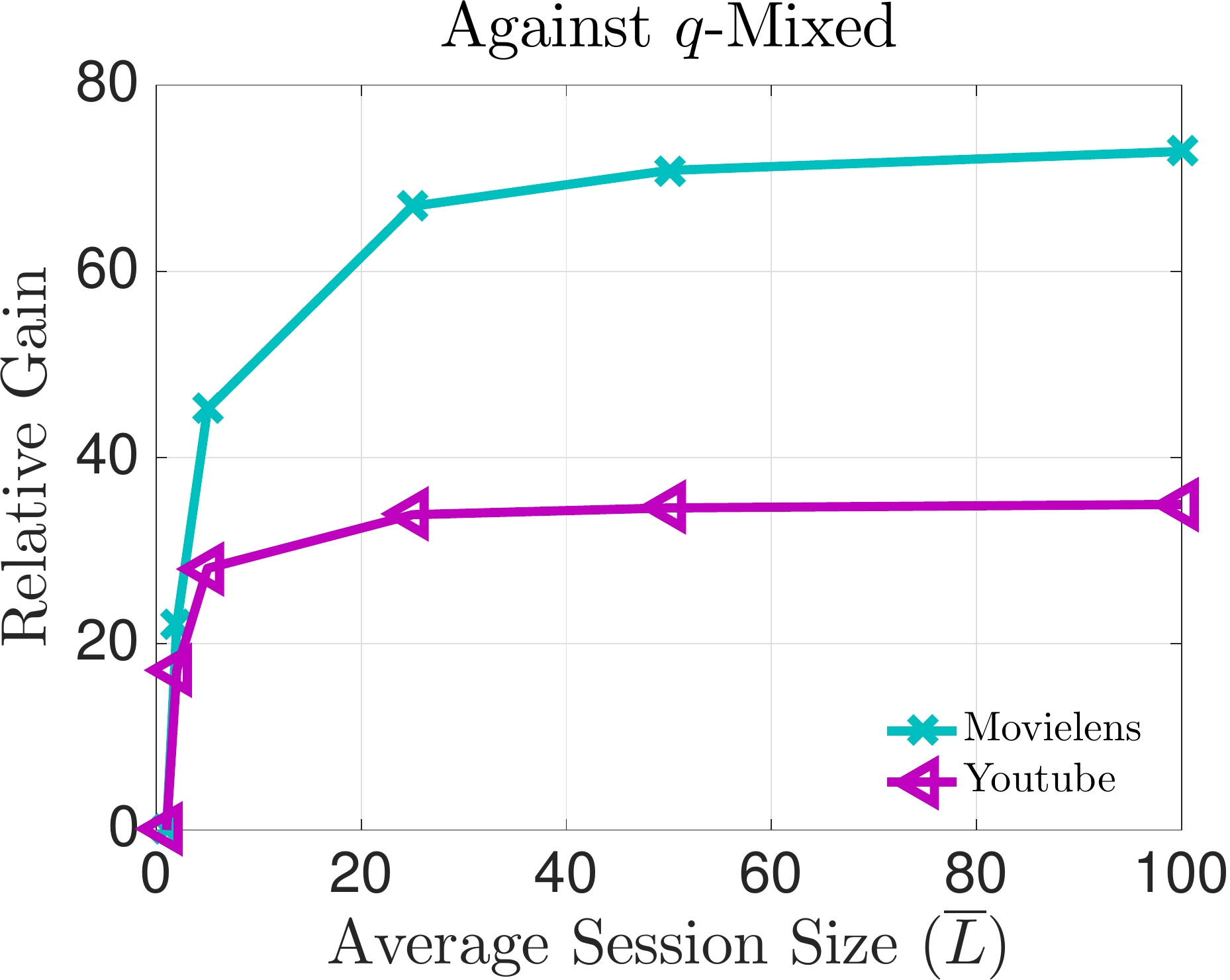}\label{fig:relative-all}}
\hspace{0.04\columnwidth}
%%%%%%%%%%%%%%%%%%%%%%%%%%%
\subfigure[Cost vs num. of recommendations $N$ and mean session size $\overline{L}$ ($\alpha = 0.75$, $q = 0.75$) - Synthetic]{\includegraphics[width=0.4\columnwidth]{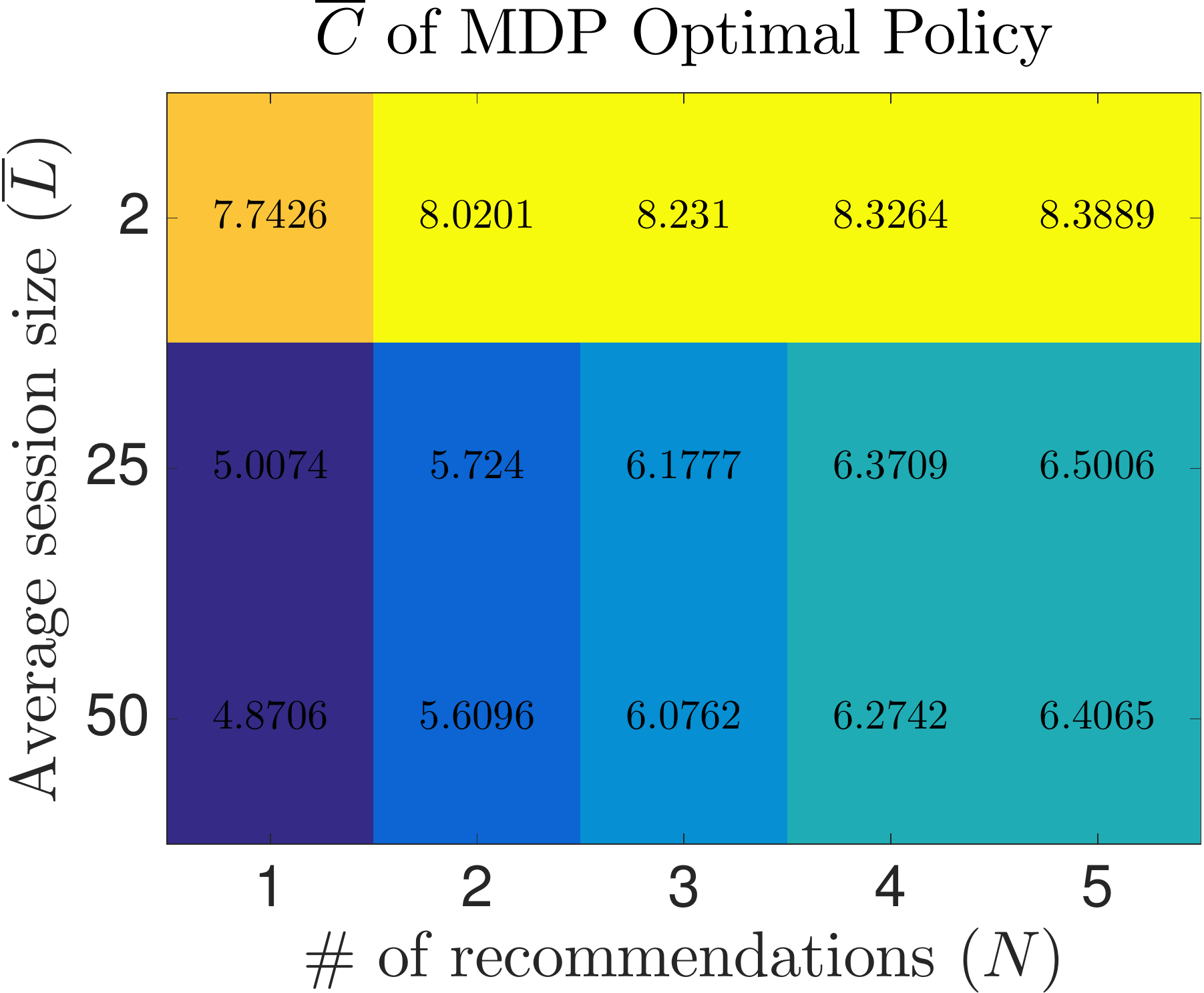}\label{fig:overN-GM}}
\caption{Subfigs. (a) - (e): Metrics vs mean session size $\overline{L}$, mvlns: $\{\alpha = 0.85, q = 0.90, N = 2\}$, yt: $\{\alpha = 0.8, q = 0.7$, N = 2$\}$}
\label{fig:general-model-all}
\end{figure}

\subsection{Results: Sensitivity Analysis}

\myitem{Effect of mean session size ($\overline{L}$).} We first compare the performance benefits of MDP which has look-ahead capabilities against myopic ones, when the size of the user session increases. To this end, in Fig. \ref{fig:general-model-all}, we vary the parameter $\lambda$ (Section \ref{sec:problem-setup}) to simulate random sessions with mean size $\overline{L} = \{1, 2, 5, 25, 50, 100\}$; we remind the reader that $\overline{L} = (1-\lambda)^{-1}$. We compare the performance of MDP against the three myopic ones discussed in Section \ref{subsec:action-space}.
% Note that from the three policies mentioned there, the only policy oriented both towards low cost and recommendation quality is the \emph{$\delta$-Mixed}, which for this simulation we tune it to be exactly $q$-\emph{Mixed}.
For the Movielens we set $q=90\%$, and for the Youtube $q=70\%$ in order to ensure high user satisfaction $\overline{Q}$ for the cost-oriented policies.
This hard constraint of $\overline{Q}$ is depicted in Figs. \ref{fig:satisfaction-movielens}, \ref{fig:satisfaction-youtube} with a dashed grey line. For the $q$-Mixed policy we set it accordingly, hence it becomes a 0.9-Mixed and a 0.7-Mixed policy. 
The extreme policy Top-$N$ achieves $\overline{Q} = 1$, which is the upper bound for any policy, and the worst $\overline{C}$. In total contrast, the Low Cost returns the best possible cost but is \emph{infeasible}. The policy $q$-Mixed offers user satisfaction at or above the feasibility boundary.

\myitem{Obs. \#1:}
%The MDP policy is aware of the session statistics and adapts $\overline{C}$ optimally to the change of the session length. Additionally, t
The MDP-optimal policy keeps the user satisfaction feasible while achieving the minimum cost $\overline{C}$, in Figs. \ref{fig:cost-movielens}, \ref{fig:cost-youtube} from the feasible myopic policies.
Moreover, in Fig. \ref{fig:relative-all}, we show the relative gain of the MDP with respect to $q$-Mixed as reference policy. Note that the respective gains against Top-$N$, which is omitted from the plot as it has no bias towards $\mathcal{C}$, are more than $1000\%$. Reasonably, the longer the horizon, the larger the gains of MDP which is equipped with look-ahead capabilities.
% achieves the best trade-off from the feasible policies (the ones that satisfy the $q$-constraint) as it has the lowest possible $\overline{C}$.

\begin{comment}
\myitem{Obs. \#2:} 
%The values of $\overline{C}$ we observe in the two datasets (see Fig. \ref{fig:general-model-all}) are quite close but are not the same. 
%
Setting $\mathbb{E}[Q_i] \ge 90\%$ in Movielens,
% a , which is higher than the YouTube, that is $\mathbb{E}[Q_i] \ge 70\%$, 
the cost values in \ref{fig:general-model-all} are not really that far apart. This hails mainly from the NFR-convenient structure of the Movielens dataset. In Movielens, the \say{average} content is more likely to have a related (high $u_{ij}$) content that is \emph{also} in $\mathcal{C}$.
%, which explains why increasing the $q$ does not affect the cost in the Movielens case. 
This can be seen in Table \ref{table:stats} where the $\text{deg}^{-}(\mathcal{C})_{mvlns} \approx 3 \times \text{deg}^{-}(\mathcal{C})_{yt}$.
\end{comment}

\myitem{A Note on Caching.} Our caching is essentially random. We could instead cache the items that have the most neighbors (in terms of $U$) or cache the top-$M$ items from the stationary distribution as created by the recommendation policy Top-$N$. Thus, we can loosely state that the cost performances of the MDP we see here, serve as a lower bound.

% contents are much better connected to the neighborhood a random content of Movielens has almost almost 3 times the $\text{deg}^{-}(\mathcal{C})$ compared to Youtube, nonetheless in absolute scale both values are low. One could falsely expect that Movielens will have lower cost than Youtube for that reason. 
%
% However, this is not the case and here is why: if we had $N = 1$, then approximately 17\% of the time the RS recommends a useful content, i.e., cached and related content, and 17\% the user \emph{would} click on it. As $N$ grows (here $N=2$) this becomes $\frac{17}{N} \%$ due to \emph{uniform selection} of our user and hence the benefit of $\text{deg}^{-}(\mathcal{C})$ fades with $N$.
% Finally, the mean $D_{\mathcal{C}}$ is not that important for our user. Essentially, low mean $D_{\mathcal{C}}$ suggests that the set $\mathcal{C}$ is very well connected through $U$ but the important aspect of our user is that she always rejects the recommendations with \emph{fixed} $1-\alpha$, no matter how good the recommendations are. Therefore, if the user finds herself in the very convenient (for the network) neighborhood $\mathcal{C}$, she can easily \say{escape} the RS well targeted actions because with fixed $1-\alpha$ she will always ignore our recommendations.

\myitem{Effect of $q$ and $\alpha$.}
%We have two key input user parameters. 
For each dataset, in Fig. \ref{fig:ds-mixed-opt-q}, we pick some $\alpha$ and tighten the quality constraint by increasing $q$. In the same fashion, in Fig. \ref{fig:ds-mixed-opt-alpha}, we pick some $q$ for every dataset and increase the value of $\alpha$.
%As both plots have been made with respect to the \emph{General Model}, 
We present here only the average cost per request, since the RS quality achieved is equal to the $q$ value selected. Furthermore, we omit the two extreme policies, Top-$N$ and Low Cost, and only compare to $q$-Mixed, who also seeks a (suboptimal) tradeoff between cost and user satisfaction.
The first thing to notice is that, for $q = 1$ and $q = 0$, the two policies coincide, which is an immediate result of Properties~\ref{prop:q=1} and \ref{prop:q=0} as the optimal policies are Top-$N$ and Low-Cost respectively.
%This is reasonable: for $q = 0$, observe that there is no constraint on the RS quality, hence both policies become the \emph{Low-Cost} one, which is the optimal one in this case; for $q = 1$, it is easy to see that both the \emph{MDP-Optimal} and the $q$-\emph{Mixed} policies degenerate to the \emph{Top-N}, which is the only feasible policy for this scenario. 
For all intermediate $q$ values, it is evident that the MDP-optimal policy improves performance compared to the $q$-Mixed.
% As expected, we see the $\overline{C}$ in extreme values being the same, which can be readily explained as in $q=0$ regardless of the number of requests, the optimal policy is \emph{LowCost-N} whereas for $q = 1$, the corresponding optimal is $Top-N$. In all the points \emph{inbetween} ($q \in (0,1)$ there is a gap in the performance of the cost. 
% We also observe that for both datasets the optimal policy which is tailored for the longer session outperforms in terms of cost the $\delta$-Mixed policy (which acts in a myopic manner).
We can also observe that the MDP-optimal policy is able to better exploit the increase of $\alpha$ as the gap between the policies becomes wider. This should not come as a surprise since the myopic policies \emph{do not} take into consideration the dynamics of user transitions. 
Note that an average session of $\overline{L}=25$, could loosely correspond to a 45min session of watching YouTube short clips~\cite{businessYoutubeSessions}.
%Note that this pair of simulations was carried out for sessions of average size $\overline{L}=25$, which roughly could translate to 45 minutes session~\cite{businessYoutubeSessions} and could be also representative for other applications such as Spotify.
\begin{figure} 
\centering
\subfigure[vs increasing $q$]{\includegraphics[width=0.4\columnwidth]{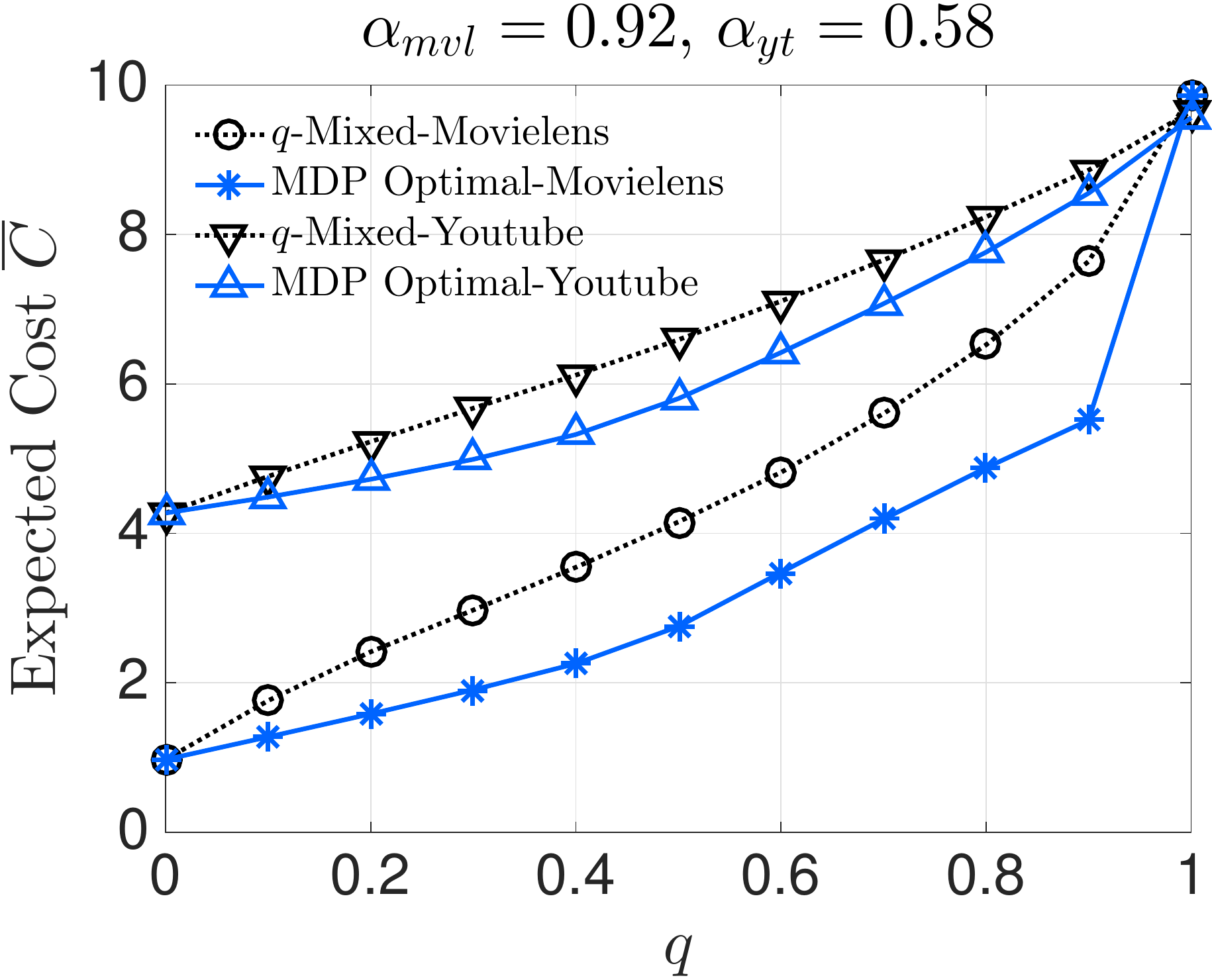}\label{fig:ds-mixed-opt-q}}
\hspace{0.08\columnwidth}
\subfigure[vs increasing $\alpha$]{\includegraphics[width=0.4\columnwidth]{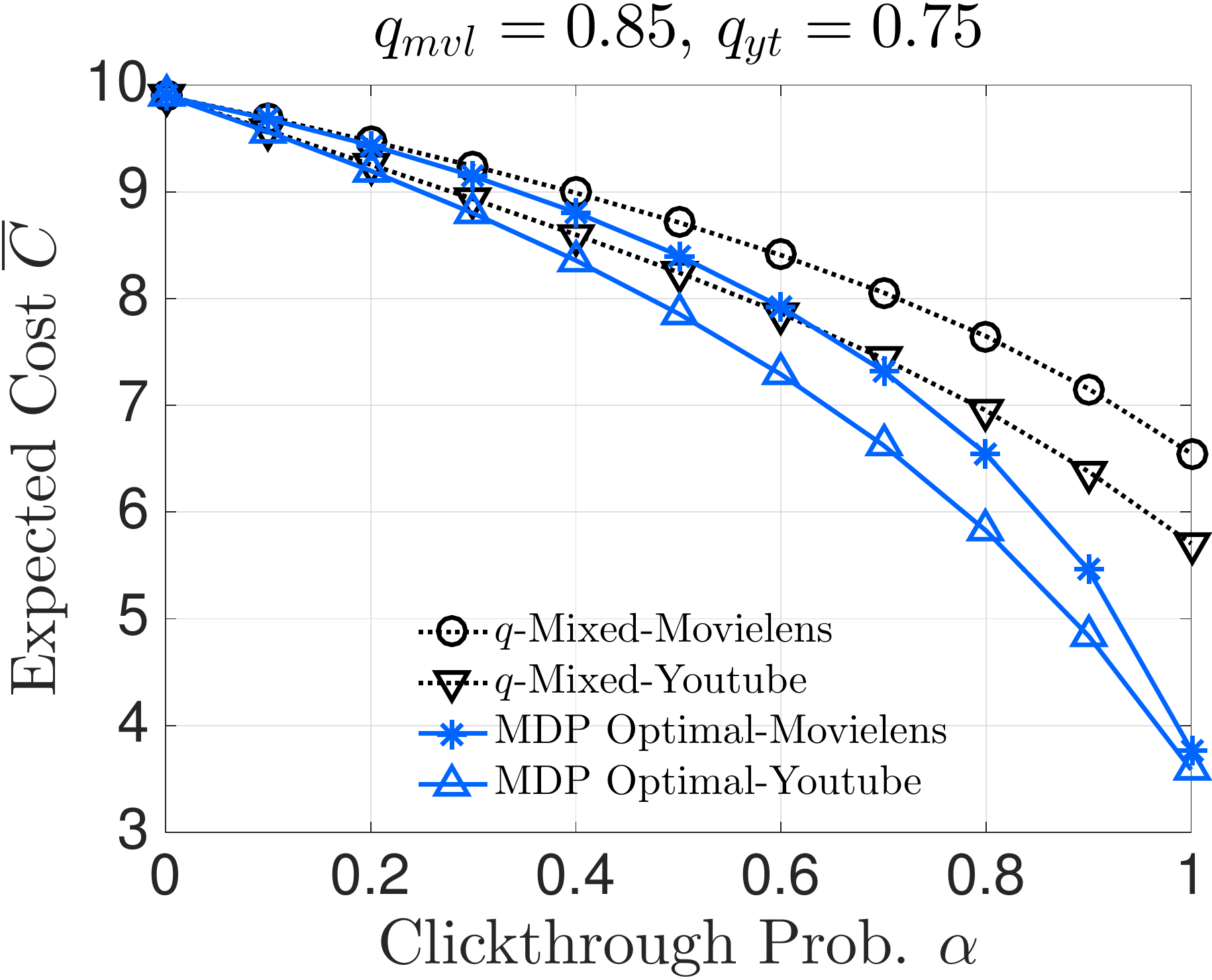}\label{fig:ds-mixed-opt-alpha}}
\caption{MDP-Optimal and $q$-Mixed evaluated on $\overline{L}=25$ requests ($N = 2$)}
\label{fig:ds-mixed-opt-model}
\end{figure}

\myitem{Effect of recommendation batch size ($N$).}
In Fig. \ref{fig:overN-GM}, we focus on the effect of $N$ on the expected cost. We provide a heatmap with increasing $\overline{L}$ on the $y$-axis and increasing $N$ on the $x$-axis. We observe that irrespective of $\overline{L}$, the cost is becoming worse with the increase of $N$ even for the Synthetic graph which has a much larger $\text{deg}^{-}(\mathcal{C}) = 0.51$.
% (even larger than Movielens), and the same effect holds, i.e., high $\text{deg}^{-}(\mathcal{C})$ stops being so important for larger $N$ for the \say{Flexible} user.

\myitem{Obs. \#3:} Network-friendly recommendations \emph{is not} an easy task, but \emph{many} network-friendly recommendations is even harder. To better grasp this, consider a myopic RS. To satisfy both parties (network and user), when the user is at content $i$, the RS must have many cached \emph{and} related contents to recommend, which by definition are less than cached \emph{or} related. 
%Thus if we request from the RS to provide only one recommendation, this is \emph{always} an easier task.
%\begin{figure}[h!]
%\centering
%\includegraphics[width=0.7\columnwidth]{overN_GM}
%\caption{Flexible user: Cost vs Num. of Recommendations ($N$) and mean Session Size $\overline{L}$ ($\alpha = 0.75$, $q = 0.75$) - Synthetic}
%\label{fig:overN-GM}
%\end{figure}
\begin{figure}
\centering
\subfigure[Proposed vs Brute Force, $\overline{L} = 50$]{\includegraphics[width=0.42\columnwidth]{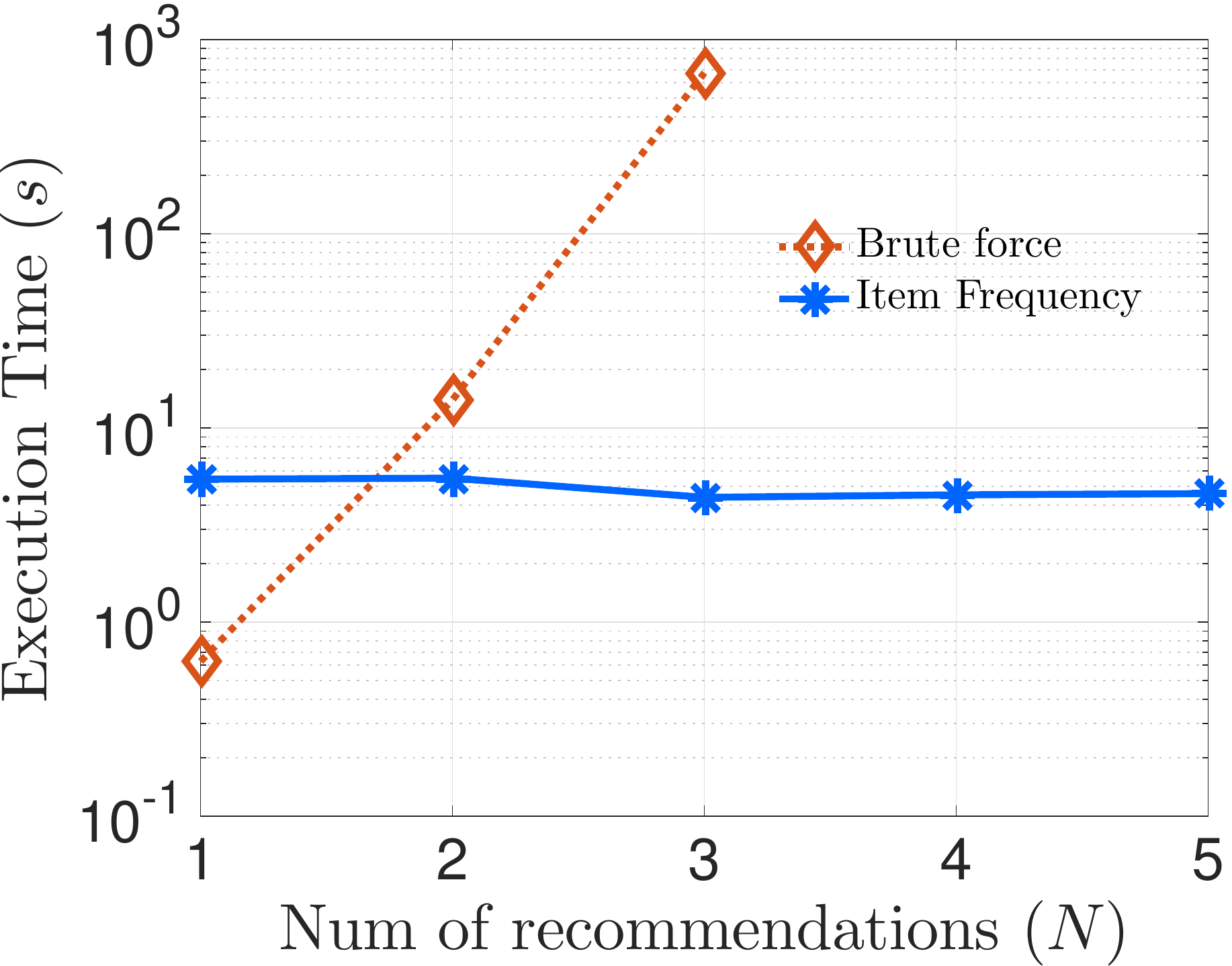}\label{fig:vanilla}}
\hspace{0.04\columnwidth}
\subfigure[Execution Times vs Library Size $K$, $\overline{L} = 20$]{\includegraphics[width=0.45\columnwidth]{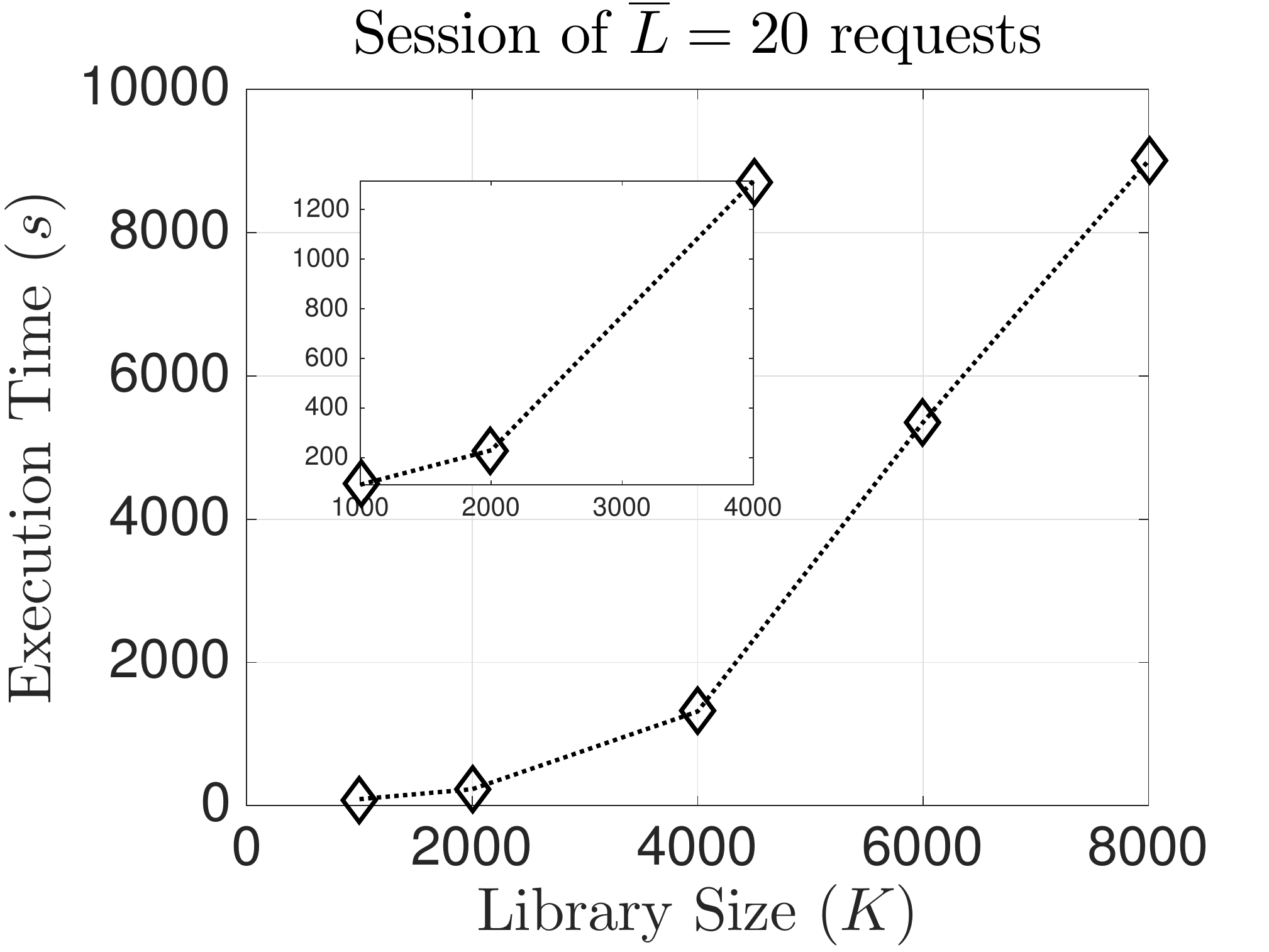}\label{ig:execution-times}}
\caption{MDP: Execution Times}
\label{ex_and_N}
\end{figure}

\subsection{Execution Time Savings (and not only)} 
Up to this point, we investigated tradeoffs between different policies. Yet an important contribution of this work is its computationally efficient framework, which we discuss next.

\myitem{Item-frequency vs Batch-frequency formulation.}
One of the main contributions of this work is that it formulates a continuous problem of item-frequencies, rather than batch-frequencies. This is profitable computationally both in the number of variables and in execution time. % on how we view the RS decisions. 
%Instead of having content batches, we have item frequencies. 
In Fig. \ref{fig:vanilla}, we choose a catalog of $K=150$, and solve the MDP using our approach and compare it against the brute force solution of a batch-MDP, which enumerates all the feasible tuples (the ones that satisfy the quality constraints), and picks the best one. As claimed in Section \ref{subsec:action-space}, we see that the increase of $N$ is devastating for the batch-frequency MDP, whereas our item-frequency approach is insensitive to it.

\myitem{Catalog Increase.}
In this part, we investigate the execution times of our item-frequency MDP algorithm, by varying the catalog size. We select $\overline{L} = 20$, increase the content library size, run the algorithm, and report the time it took until completion. These results show that our MDP can tune recommendations of practical size and not only toy scenarios of some hundred items. As stated in~\cite{femto}, even a library of $K = 1000$ can be considered practical since it could refer to the 1000 most popular files of Netflix for example. The authors in~\cite{chatzieleftheriou2019TMC} perform simulation with sizes $K = 1000$ and $K = 10000$, which is the same order as our experiments; however, they do not report any execution time results. The MDP needs about 9000 seconds ($\approx$ 2.5h) for a library of 8000, using the 8Gb RAM PC, and under-exploited parallelisation. These run-times will be significantly decreased in a powerful server with multiple cores, as the Policy Iteration algorithm we have implemented runs on as many cores as it finds available.
% hat for 2K and 1K library the completion time for the optimal policy is simply a few minutes. 
% On the other hand, for the rational user we were able to test our algorithm even for an 8K library, and impressively the runtime for the reasonably \say{long} session of $\overline{L} = 10$ contents on average takes about 10 minutes.

%\myitem{Take-away:} The MDP framework offers a natural decomposition to the problem to solve the Network-Friendly recommendations problem \emph{for any} mean session size.
%An important advantage of the algorithm we propose is its low complexity.
Here we compare our MDP with the policy of \cite{giannakas2019wiopt} where the objective is to minimize the per request average cost over an \emph{infinite size} session. Their framework easily reduces to ours by setting $\lambda \to 1$ (in practice we set $\lambda = 0.9$, that is 10 steps look-ahead) and assuming $\alpha_{ij} = \frac{\alpha}{N}$, i.e., uniform click towards recommended content. 
%For completeness, we also report the performance of the myopic $q$-Mixed policy under the infinite session regime.
%; the policy was found using our MDP with $\lambda = 0$ for the \say{Flexible} user. 
In~\cite{giannakas2019wiopt}, the authors formulate the average cost minimization as an LP of size $K^2$ and the optimal solution is found using CPLEX. Their solution is constrained to obey \emph{stationarity} which builds a very demanding set of constraints and is unrealistic as the size of a user session \emph{is finite} in practice. In Table \ref{table:mdp-wowmom}, for the two datasets, we report the execution time of the algorithm and the achieved mean cost $\overline{C}$ under the stationary regime, i.e., we plug our policy into the objective of \cite{giannakas2019wiopt}. In that table, we refer to our policy as MDP(0.9) (due to the selected $\lambda$) and to the one of \cite{giannakas2019wiopt} as \emph{OPT}.
% However there, (a): instead of breaking the problem in smaller more easily solvable problems, CPLEX attempts to solve the whole problem in one go, which is a typical situation of \say{the whole is greater than the sum of its parts} and (b): the solution has the unnecessary extra constraints of stationarity which make the optimization problem even harder.
%\begin{table}
%\centering
%\caption{MDP vs S-o-A (under \cite{giannakas2018show}) }\label{table:mdp-wowmom}
%\begin{small}
%\begin{tabular}{l|c|c|c|c}
%{} & \multicolumn{2}{c|}{Objective Value} & \multicolumn{2}{c}{Execution Time ($s$)}\\
%\hline
%{Dataset/Method}                       			&{MDP(0.99)} &{\cite{giannakas2018show}} &{MDP(0.99)} &{\cite{giannakas2018show}}\\
%\hline \hline
%{Movielens}			            			&{5.696} &{5.543} &{86.40} &{1065.78}\\
%\hline
%{Youtube}		                					&{6.161} &{6.100} &{148.49} &{5164.69}\\
%% \hline
%% {Synthetic}		                				&{4.205} &{4.122} &{50.61} &{622.57} \\
%\end{tabular}
%\end{small}
%\end{table}
%\begin{table}
%\centering
%\caption{MDP vs S-o-A (under \cite{giannakas2018show}) }\label{table:mdp-wowmom}
%\begin{small}
%\begin{tabular}{l|c|c|c|c|c|c}
%{} & \multicolumn{3}{c|}{Objective Value} & \multicolumn{2}{c}{Execution Time ($s$)}\\
%\hline
%{Dataset/Method}                       			&{$q$-Mix} &{MDP(0.99)} &{OPT} &{$q$-Mix} &{MDP(0.99)} &{OPT}\\
%\hline \hline
%{Movielens}			            			&{7.649} &{5.696} &{5.543} &{16.85} &{86.40} &{1065.78}\\
%\hline
%{Youtube}		         					&{7.022} &{6.161} &{6.100} &{39.98} &{148.49} &{5164.69}\\
%\end{tabular}
%\end{small}
%\end{table}
\vspace{5mm}
\begin{table}[h!]
\centering
\caption{MDP(0.99) vs \emph{OPT} (under \cite{giannakas2019wiopt})}\label{table:mdp-wowmom}
\begin{small}
\begin{tabular}{c|c|c|c|c}
{} & \multicolumn{2}{c|}{Cost (units)} &\multicolumn{2}{|c} {Exec. Time ($s$)}\\
\hline
{} & {MDP(0.9)} &{\emph{OPT}} &{MDP(0.9)} &{\emph{OPT}}\\
\hline
{Movielens}	&{5.5625} &{5.5432} &{105} &{560}\\
\hline
{Youtube} &{6.1005} &{6.1001} &{253} &{1997}\\
\end{tabular}
\end{small}
\end{table}
The results of this experiment are summarized in Table~\ref{table:mdp-wowmom}. Impressively, there is an execution time speed up by a factor of 5 and 10 for the two datasets, while sacrificing almost nothing in terms of cost performance.

We now do the exact opposite; for smaller sessions, $\overline{L} = \{1,2,3,4,5\}$, we present the relative gain of MDP over the policy of \cite{giannakas2019wiopt} and the $q$-Mixed. 
Our approach finds the optimal cost for all $\overline{L}$. The smaller the horizon, the bigger the gain of MDP with respect to~\cite{giannakas2019wiopt}. Reasonably, as the horizon increases, the relative gain fades as~\cite{giannakas2019wiopt} is exactly tailored for \emph{very} long sessions.
% , and reasonably the smaller the horizon, the bigger the relative gain of MDP 
% and as we can observe, as $\overline{L}$ grows, reasonably, the relative gain with respect to \cite{giannakas2019wiopt} decreases as the MDP can control its horizon and starts resembling~\cite{giannakas2019wiopt}.
Note that for such small sessions the MDP rutime is obviously even lower than the one shown in Table~\ref{table:mdp-wowmom}, because smaller $\overline{L}$ translates to smaller $\lambda$, which implies faster convergence of the policy iteration. Finally, as seen in previous plots, the gain over the $q$-Mixed is growing with the horizon $\overline{L}$. 
% Note that MDP does find the optimal decisions given these $\overline{L}$ whereas for $\overline{L} < \infty$, the approach of \cite{giannakas2019wiopt} is suboptimal. Reasonably, as \cite{giannakas2019wiopt} has unnecessarily large vision (i.e., to infinity), the smaller the $\overline{L}$, the larger the relative gain of MDPs. Additionally, the execution of \cite{giannakas2019wiopt} needed 1065.7 seconds for the simulated scenario (see Table \ref{table:mdp-wowmom}). These results reveal that MDP could easily tune a RS optimally and more efficiently regardless the mean session size.
%
% W when we need the RS decisions for sessions of  items, where MDP is optimal as it can select the right $\lambda$ for any value in $\overline{L}$. Here the approach of infinite item session of \cite{giannakas2019wiopt} needs 1065.7 seconds to return the policy and is cost-wise suboptimal. Plausibly, for smaller horizons the approach of~\cite{giannakas2019wiopt} has unnecessary vision which negatively affects the expected cost. Obviously, as the horizon becomes larger, the two methods will eventually converge.
% \begin{table}[h!]
% \centering
% \caption{Evaluating~\cite{giannakas2019wiopt} on smaller $\overline{L}$, Movielens}\label{table:mdp-wowmom2}
% \begin{small}
% \begin{tabular}{c|c|c|c|c|c}
% {$\overline{L}$} &{1} &{2} &{3} &{4} &{5}\\
% \hline
% {Relative Gain}	&{33.36} &{19.26} &{18.66} &{9.12} &{2.66}\\
% \hline
% {Time MDP ($s$)} &{47.4} &{49.0} &{44.0} &{55.36} &{44.95}\\
% \end{tabular}
% \end{small}
% \end{table}
\vspace{5mm}
\begin{table}[h!]
\centering
\caption{Evaluating on smaller $\overline{L}$, Movielens}\label{table:mdp-wowmom2}
\begin{small}
\begin{tabular}{c|c|c|c|c|c}
{Gain over} &{$\overline{L}$=1} &{$\overline{L}$=2} &{$\overline{L}$=3} &{$\overline{L}$=4} &{$\overline{L}$=5}\\
\hline \hline
{\cite{giannakas2019wiopt} ($\%$)}	&{21.29} &{16.27} &{10.4528} &{4.48} &{4.87}\\
\hline
{$q$-Mix ($\%$)} &{0.01} &{9.67} &{21.46} &{22.39} &{27.69}\\
\end{tabular}
\end{small}
\end{table}

% relOPT =

%   21.2963   16.2745   10.4528    4.4875    4.8749

% relQ =

%   52    9.6732   21.4645   22.3902   27.6963

% relOPT =

%   29.3774   14.6923   13.4027    6.0894    8.8355    2.0909

% relQ =

%     0.7782   11.5156   21.9526   24.2626   28.6815   29.7949

%Finally notice that the $q$-Mixed policy has even lower execution time than both and is quite lightweight, however in terms of cost it is heavily suboptimal result in both datasets.

%%%%%%%%%%%%%%%%%EEEEEENNNNNNNDDDDDDD%%%%%%%%%%%%%%%%%%%%%%%%%%%%%%%%%%%%%%%%%%%%%%%%%%%%%%%%%%%%%%%%%%%%%%%%%%%%%%%%%%%%%%%%%%%%%%%%%%%%%%%%%%%%%%%%%%%%%%%%%%%%%%%%%%%%%%%%%%%%%%%%%%%%%%%%%%%%%%%%%%%%%%%%%%%%%%%%%%%%%%%%%%%%%%%%%%%%%%%%%%%%%%%%%%%%%%%%%%%%%%%%%%%%%%%%%%%%%%%%%%%%%%%%%%%%%%%%%%%%%%%%%%%%%%%%%%%%%%%%%%%%%%%%%%%%%%%%%%%%%%%%%%%%%%%%%%%%%%%%%%%%%%%%%%%%%%%%%%%%%%%%%%%%%%%%%%%%%%%%%%%%%%%

\begin{comment}
Evidently, the runtimes of the \say{Flexible} user are an order of magnitude higher than the algorithm need for the \say{Peaky} user case. This is reasonable as the minimizers of the former are the $K$-sized LPs (see \nameref{problem:flat}) whereas the ones of the latter are simply sorting operations (see \nameref{problem:peaky}) and as such they naturally scale better with the increase of $K$. This hints that the algorithm for the family of users that resemble the \say{Peaky} user, could very well be applicable in a much larger scenario.
\end{comment}
% \input{sims_v2}

%% Section 6
\section{Conclusions}
\label{sec:conclusions}
%We employed the model-based optimization road and did so in a problem where the human factor is heavily involved. However, the MDP sets the stage for learning based techniques, e.g., Q-learning, where the RS is trained over massive user request datasets; an approach which requires minimal assumptions on the user and could thus generalize better than model-specific methods. Importantly, the use of item frequencies (instead of batches) can help us dramatically reduce the action space set in this framework as well.

We have developed a very promising MDP framework for optimal look-ahead NFR that is able to exploit the structure of the content graph and discover non-obvious recommendations. 
%The framework is flexible and can incorporate several possible user behaviours related to the user's clickthrough.
More importantly, by using item-frequency recommendations in the Bellman equations we have proposed an algorithm that scales well both with the size of content library, as well as with the batch-size. The complexity remains low as the inner optimization problems have less unknowns, they are linear or at worse convex, and allow for parallelisation. 
Finally, as MDP sets the stage for learning-based approaches, we firmly believe that our reduced variable representation MDP (via the item-frequency) can significantly speed up the training phase of such algorithms.

\bibliographystyle{ieeetr}
\bibliography{raporti}

% \section{TODO LIST}
% \input{comments}

\end{document}